\documentclass[11pt]{article}
\PassOptionsToPackage{x11names}{xcolor}
\usepackage[T1]{fontenc}
\usepackage{soul} 
\usepackage{phfqit}
\usepackage{proba} 
\usepackage{amssymb}
\usepackage{amsmath}
\usepackage{amsthm}
\usepackage{amsfonts}
\usepackage{mathtools}
\usepackage{xspace}
\usepackage{bm}
\usepackage{nicefrac}
\usepackage{commath}
\usepackage{bbm}
\usepackage{boxedminipage}
\usepackage{xparse}
\usepackage{xcolor}
\usepackage{float}
\usepackage{multirow}
\usepackage{graphicx}
\usepackage{caption}
\usepackage{subcaption}
\usepackage{ifthen}
\usepackage{adjustbox}
\usepackage[ruled,vlined,linesnumbered]{algorithm2e}
\usepackage{colortbl} 
\usepackage{fancyhdr}   
\usepackage{hyperref}
    \hypersetup{colorlinks=true, linkcolor=blue, filecolor=magenta, urlcolor=red, citecolor=blue, breaklinks}
\usepackage{fullpage}
\usepackage[utf8]{inputenc}
\usepackage{environ}
\usepackage[colorinlistoftodos]{todonotes}
\usepackage{mdframed}
\usepackage{setspace}

\usepackage[capitalize]{cleveref}

\usepackage{tikz}
\usetikzlibrary{shapes}
\usetikzlibrary{positioning}
\usetikzlibrary{fit}
\usetikzlibrary{calc}
\usetikzlibrary{backgrounds}
\usetikzlibrary{patterns}
\usetikzlibrary{matrix}
\usetikzlibrary{arrows}

\usepackage[noadjust,nospace]{cite}

\newtheorem{proposition}{Proposition}
\newtheorem{lemma}{Lemma}
\newtheorem{theorem}{Theorem}

\newtheorem{definition}{Definition}

\newtheorem{claim}{Claim}
\theoremstyle{remark}
\newtheorem{remark}{Remark}

\crefname{claim}{Claim}{Claims}
\Crefname{claim}{Claim}{Claims}

\newcommand{\ie}{\text{i.e.}\xspace}

\newcommand{\eg}{\text{e.g.}\xspace}

\newcommand{\polylog}{\ensuremath{\mathrm{polylog}}\xspace}

\DeclareMathOperator*{\Expect}{\mathbb{E}}

\newcommand{\eps}{\ensuremath{\varepsilon}}

\def\*#1{\mathbf{#1}}
\def\+#1{\mathcal{#1}}

\NewEnviron{problem}[1]{%
	\begin{center}\fbox{\parbox{6in}{%
				{\centering\scshape #1\par}%
				\parskip=1ex
				\everypar{\hangindent=1em}%
				\BODY
}}\end{center}}

    \newcommand{\cD}{\ensuremath{{\mathcal D}}\xspace}

    \newcommand{\cL}{\ensuremath{{\mathcal L}}\xspace}

    \newcommand{\Dec}{\ensuremath{\mathsf{Dec}}\xspace}

    \newcommand{\ignore}[1]{{}}

\def\final{0} 
\ifnum\final=0  
\newcommand{\authnote}[3]{\textcolor{#2}{{\sf (#1's Note: {\sl{#3}})}}}
\newcommand{\jeremiah}{\authnote{Jeremiah}{blue}}
\newcommand{\minshen}{\authnote{Minshen}{purple}}
\newcommand{\xnote}{\authnote{Xin}{magenta}}
\newcommand{\enote}{\authnote{Elena}{green}}
\newcommand{\knote}{\authnote{Kuan}{cyan}}
\newcommand{\yu}{\authnote{Yu}{cyan}}
\newcommand{\alex}[1]{\authnote{Alex}{orange}{#1}}
\else 
\newcommand{\authnote}{}
\newcommand{\jeremiah}[1]{}
\newcommand{\minshen}[1]{}
\newcommand{\xnote}[1]{}
\newcommand{\enote}[1]{}
\newcommand{\knote}[1]{}
\newcommand{\yu}[1]{}
\newcommand{\alex}[1]{}
\fi

\newif\ifnotes\notestrue 

\usepackage[utf8]{inputenc}
\usepackage{thm-restate}
\usepackage{orcidlink}

\usepackage[
    final 
]{changes}
\definechangesauthor[name=Reviewer 3, color=Purple2]{R3}
\definechangesauthor[name=Reviewer 2, color=red!75!black]{R2}
\definechangesauthor[name=Reviewer 1, color=green!60!black]{R1}

\title{Exponential Lower Bounds for 2-query Relaxed Locally Decodable Codes\thanks{An extended abstract of this paper appeared in the \href{https://doi.org/10.4230/LIPIcs.CCC.2023.14}{Proceedings of the 38th Computational Complexity Conference, 2023}~\cite{conf-version}.}}  

\author{%
    Alexander R. Block\orcidlink{0000-0002-2632-763X}\thanks{University of Illinois at Chicago. \texttt{arblock@uic.edu}.}%
    \and%
    Jeremiah Blocki\orcidlink{0000-0002-5542-4674}\thanks{Purdue University. \texttt{jblocki@purdue.edu}. Supported by NSF CCF-1910659 and NSF CAREER Award CNS-2047272.}%
    \and%
    Kuan Cheng\orcidlink{0000-0002-8972-1749}\thanks{Peking University. \texttt{ckkcdh@pku.edu.cn}. Supported by the National Natural Science Foundation of China under Grant 62472008 and by CCF-Huawei Populus Grove Fund (CCF-HuaweiLK2025005).}%
    \and%
    Elena Grigorescu\orcidlink{0000-0001-9673-4313}\thanks{University of Waterloo. \texttt{elena-g@uwaterloo.ca}. Supported by {NSF CCF-1910659, NSF CCF-1910411, and NSF CCF-2228814}, while at Purdue University.}%
    \and%
    Xin Li\orcidlink{0000-0002-9408-2451}\thanks{Johns Hopkins University. \texttt{lixints@cs.jhu.edu}. Supported by NSF CAREER Award CCF-1845349 and NSF Award CCF-2127575.}%
    \and
    Yu Zheng\thanks{\texttt{hizzy1027@gmail.com}. Supported by {NSF CAREER Award CCF-1845349}, while at Johns Hopkins University.}%
    \and%
    Minshen Zhu\orcidlink{0000-0003-1927-6085}\thanks{\texttt{minshen.zh@gmail.com}. Supported by {NSF CCF-1910659, NSF CCF-1910411, and NSF CCF-2228814}, while at Purdue University.}%
}

\begin{document}
\maketitle

\begin{abstract}
    Locally Decodable Codes (LDCs) are error-correcting codes $C\colon\Sigma^n\rightarrow \Sigma^m,$ encoding \emph{messages} in $\Sigma^n$ to \emph{codewords} in $\Sigma^m$,  with superfast decoding algorithms. 
    They are important mathematical objects in many areas of theoretical computer science, yet the best constructions so far have codeword length $m$ that is super-polynomial in $n$, for codes with constant query complexity and constant alphabet size\added[id=R1, comment={(1)}]{, while tolerating a constant-fraction of errors}.
  
    In a very surprising result, Ben-Sasson, Goldreich, Harsha, Sudan, and Vadhan (SICOMP 2006) show how to construct a relaxed version of LDCs (RLDCs) with constant query complexity and \replaced[id=R1, comment={(2) Language here; does ``almost linear'' mean $m=n^{1+o(1)}$?}]{slightly super-linear}{almost linear} codeword length over the binary alphabet, and used them to obtain significantly-improved constructions of Probabilistically Checkable Proofs. 
  
    \deleted[id=R1, comment={(3)}]{In this work, we study RLDCs in the standard Hamming-error setting.}
    \replaced[id=R1, comment={(3)}]{In this work, we }{We }prove an exponential lower bound on the length of Hamming RLDCs making $2$ queries (even adaptively) over the binary alphabet. 
    This answers a question explicitly raised by Gur and Lachish (SICOMP 2021) and is the first exponential lower bound for RLDCs. 
    Combined with the results of Ben-Sasson et al., our result exhibits a ``phase-transition''-type behavior on the codeword length for some constant-query complexity. 
    We achieve these lower bounds via a transformation of RLDCs to \deleted[id=R1, comment={(3)}]{standard Hamming} LDCs, using a careful analysis of restrictions of message bits that fix codeword bits.
\end{abstract}

\clearpage
\section{Introduction}
Locally Decodable Codes (LDCs) \cite{KatzT00, SudanTV99} are error-correcting codes $C\colon \Sigma^n \rightarrow \Sigma^m$ that have super-fast decoding algorithms that can recover individual symbols of a \emph{ message} $\mathbf{x}\in \Sigma^n$, even when worst-case errors are introduced in the \emph{codeword} $C(\mathbf{x})$. 
Similarly, Locally Correctable Codes (LCCs) are error-correcting codes $C\colon \Sigma^n \rightarrow \Sigma^m$ for which there exist very fast decoding algorithms that recover individual symbols of the \emph{codeword} $C(\mathbf{x})\in \Sigma^m$, even when worst-case errors are introduced. 
LDCs/LCCs were first \replaced[id=R3, comment={(1)}]{defined and formally studied }{discovered}\comment[id=R1]{(4) Rephrased attribution.} by Katz and Trevisan \cite{KatzT00} and since then have proven to be crucial tools in many areas of computer science, including private information retrieval, probabilistically checkable proofs, self-correction,  fault-tolerant circuits, hardness amplification, and data structures (see, \eg, \cite{BabaiFLS91,LundFKN92,BlumLR93,BlumK95,ChorKGS98,ChenGW13,AndoniLRW17} and surveys \cite{Tre04-survey,Gasarch04}).

The \emph{parameters} of interest of these codes are their \emph{rate}, defined as the ratio between the message length $n$ and the codeword length $m$, their \emph{relative minimum distance}, defined as the minimum normalized Hamming distance between any pair of codewords, and their \emph{locality} or \emph{query complexity}, defined as the number of queries a decoder makes to a received word $\mathbf{y}\in \Sigma^m$. 
Trade-offs between the achievable parameters of LDCs/LCCs have been studied extensively over the last two decades \cite{KerenidisW04,WehnerW05,GoldreichKST06,Woodruff07,Yekhanin08,Yekhanin12,DvirGY11,Efremenko12,GalM12,BhattacharyyaDS16,BhattacharyyaG17,bhattacharyya2017lower,DvirSW17,KoppartyMRS17,BhattacharyyaCG20} (see also surveys by Yekhanin \cite{Yekhanin12} and Kopparty and Saraf \cite{KoppartyS16}). 

Specifically, for $2$-query Hamming LDCs/LCCs, it is known that $m=2^{\Theta(n)}$ \cite{KerenidisW04,GoldreichKST06,Ben-AroyaRW08,bhattacharyya2017lower}. 
However, for $q>2$ queries, the current gap between upper and lower bounds is super-polynomial in $n$: the best LDC/LCC constructions have super-polynomial codeword length \cite{Yekhanin08,DvirGY11,Efremenko12}, while the most general lower bounds for locality $q\geq 3$ are of the form $m=\Omega((\frac{n}{\log n})^{1+1/(\lceil\frac{q}{2}\rceil-1)})$ \cite{KatzT00,KerenidisW04}. 
In particular, for $q=3$, Katz and Trevisan \cite{KatzT00} showed an $m=\Omega(n^{3/2})$ lower bound, which was improved by Kerenidis and de Wolf\deleted[id=R3, comment={(2)}]{e}\comment[id=R1]{(5)} \cite{KerenidisW04} to $m=\Omega(n^2/\log^2 n)$. 
This was further improved by Woodruff \cite{Woodruff07,Woodruff12} to $m=\Omega(n^2/\log n)$ for general codes and $m=\Omega(n^2)$ for linear codes. 
Bhattacharyya, \replaced[id=R1, comment={(6)}]{Chandran, and Ghoshal}{Gopi, and Tal} \cite{BhattacharyyaCG20} used new combinatorial techniques to obtain the same $m=\Omega(n^2/\log n)$ 
bound. 
\replaced{Recently, Alrabiah et al.}{A very recent paper} \cite{AlrabiahGKM} breaks the quadratic barrier and proves that  $m=\Omega(n^3/\polylog(n))$.
We note that the exponential lower bound on the length of $3$-query LDCs from \cite{GalM12} holds only for some restricted parameter regimes, and \replaced[id=R2]{does }{do }not apply to the natural ranges of the known upper bounds.
\added[id=R3, comment={(3)}]{Even more recently, Basu et al. \cite{BasuHKL25} and Janez and Manohar \cite{JanzerM25} prove that $m=\Omega((\frac{n}{\log n})^{1+1/(\frac{q}{2}-1)})$ for all odd $q\geq 3$, matching the lower bounds for even $q$.}

Motivated by this large gap in the constant-query regime, as well as by applications in constructions of Probabilistically Checkable Proofs (PCPs), Ben-Sasson, Goldreich, Harsha, Sudan, and Vadhan \cite{Ben-SassonGHSV06} introduced a relaxed version of LDCs for Hamming errors. 
Specifically, the decoder is allowed to output a ``decoding failure'' answer (marked as ``$\bot$''), as long as it errs with some small probability. 
More precisely, a \emph{\added[id=R1]{binary }$(q,\delta, \alpha, \rho)$-relaxed LDC} is an error-correcting code satisfying the following properties.

\begin{definition}\label{def:strongRLDC} A $(q,\delta, \alpha\deleted[id=R1]{, \rho})$-Relaxed Locally Decodable Code \replaced[id=R1]{$C \colon \{0,1\}^{n} \rightarrow \{0,1\}^{m}$}{${C}: \Sigma^n \rightarrow \Sigma^m$} is a code for which there exists a decoder that makes at most $q$ queries to the received word $\mathbf{y} \in \{0,1\}^{m}$, and satisfies the following\deleted{ further} properties:
\begin{enumerate} 
    \item (Perfect completeness.) For every $i\in [n]$, if $\mathbf{y}=C(\mathbf{x})$ for some message $\mathbf{x}$, then the decoder, on input $i$, outputs $x_i$ with probability $1$.\footnote{We remark that the initial definition in \cite{Ben-SassonGHSV06} only requires that $x_i$ is output with probability $2/3$ when there are no errors. However, to the best of our knowledge, all constructions of RLDCs (and LDCs) from the literature have perfect completeness. Moreover, some lower bounds (e.g., \cite{bhattacharyya2017lower}
    ) only hold with respect to perfect completeness.} 
    \item (Relaxed decoding.) For every $i\in [n]$, if \replaced[id=R1]{$\mathbf{y} \in \{0,1\}^{m}$}{$y$} is such that \replaced[id=R1]{$\mathsf{HAM}(\mathbf{y},C(\mathbf{x})) \leq \delta m$ }{$\mathrm{dist}(y, C(x))\leq \delta$ }for some unique $C(\mathbf{x})$\added[id=R1]{, where $\mathsf{HAM}$ denotes the Hamming distance}, then the decoder, on input $i$, outputs $x_i$ or $\bot$ with probability $\geq \alpha$. 
\end{enumerate}
    \deleted[id=R1, comment={(9, 10)}]{We define an RLDC that satisfies all $3$ conditions to be a \emph{strong} RLDC, and one that satisfies just the first $2$ conditions to be a \emph{weak} RLDC, in which case it is called a $(q,\delta, \alpha)$-RLDC.} 
    \replaced{If}{Furthermore, if} the $q$ queries are made in advance, before seeing entries of the \replaced{received word}{codeword}, then the decoder is said to be \emph{non-adaptive}; otherwise, it is called \emph{adaptive}.
\end{definition}

\deleted[id=R1]{%
The above definition is quite general, in the sense that $\mathrm{dist}(a,b)$ can refer to several different distance metrics. 
In the most natural setting, we use $\mathrm{dist}(a,b)$ to mean the ``relative'' Hamming distance between $a,b\in \Sigma^m$, namely $\mathrm{dist}(a,b)=|\{i\colon a_i\ne b_i\}|/m$. 
This corresponds to the standard RLDCs for Hamming errors. 
Throughout this paper, we only consider the case where $\Sigma=\{0,1\}$.
}
\Cref{def:strongRLDC} has also been recently extended to the notion of \emph{Relaxed Locally Correctable Codes (RLCCs)} by Gur, Ramnarayan, and Rothblum \cite{GurRR20}. 
RLDCs and RLCCs have been studied in a sequence of exciting works, where new upper and lower bounds have emerged, and new applications to probabilistic proof systems have been discovered \cite{ChiesaGS20, GurRR20, AsadiS21, GurL21, CohenY22}. 

Surprisingly, Ben-Sasson et al. \cite{Ben-SassonGHSV06} construct strong RLDCs with $q=O(1)$ queries and  \replaced[id=R3, comment={(4)}]{$m = n^{1+O(1/q)}$}{$m=n^{1+O(1/\sqrt{q})}$},\footnote{\added[id=R3]{The original bound of \cite{Ben-SassonGHSV06} was $n^{1+O(1/\sqrt{q})}$, which was improved to $n^{1+O(1/q)}$ by Asadi and Shinkar \cite{AsadiS21}. However, Goldreich \cite{Goldreich25} recently remarked that the original bound was due to a trivial oversight, and the actual bound obtained in \cite{Ben-SassonGHSV06} was indeed $n^{1+O(1/q)}$.}} 
\deleted[id=R3, comment={(4)}]{ and more recently Asadi and Shinkar \cite{AsadiS21} improve the bounds to $m=n^{1+O(1/q)}$,} in stark contrast with the state-of-the-art constructions of standard LDCs. 
Gur and Lachish \cite{GurL21} show that these bounds are \replaced[id=R1, comment={(15)}]{nearly }{in fact }tight, as for every $q\geq 2$, every weak $q$-query RLDC must have length  $m=n^{1+1/O(q^{2})}$ for non-adaptive decoders\added[id=R1, comment={(14)}]{, and a weaker bound of $m = n^{1+1/O(2^{q}\log^{2}(q))}$ for adaptive decoders}. 
We remark that the lower bounds of \cite{GurL21} hold even when the decoder does not have perfect completeness and in particular valid message bits are decoded with success probability $2/3.$ 
Dall'Agnon, Gur, and Lachish \cite{DallAgnolGL23} further extend these bounds to the setting where the decoder is adaptive, with $m=n^{1+1/O(q^{2}\log^2 q)}.$ 
\added[id=R3, comment={(4)}]{Very recently, Goldberg, Gur, and Saraogi \cite{GoldbergGS26} give a lower bound of $m=n^{1+1/O(q)}$ for non-adaptive RLDCs, and then use a reduction due to Goldberg \cite{Goldberg24} to obtain the same bound for linear RLDCs, nearly matching the upper bound of \cite{Ben-SassonGHSV06}, and conjecture that the machinery of \cite{DallAgnolGL23} can extend their results to (non necessarily linear) adaptive RLDCs with a matching bound.}
\comment[id=R2]{Add \cite{GoldbergGS26}.}
\comment[id=R1]{R1 also adds this comment.}

\subsection{Our Results}
As discussed, since the introduction of RLDCs, unlike standard LDCs, they have displayed a behavior amenable to \replaced[id=R1, comment={(16) same comment as (R1, 2).}]{slightly super-linear}{nearly linear-size} constructions, with almost matching upper and lower bounds. 
However, Gur and Lachish \cite{GurL21} recently conjectured that for $q=2$ queries, there is in fact an exponential lower bound, matching the bounds for standard LDCs. 
 
In this paper, we answer this question affirmatively and prove their conjecture; namely, we show that Hamming $2$-query RLDCs require exponential length, for both adaptive and non-adaptive decoders. 
\begin{theorem}\label{thm:main}
    Let $C \colon \{ 0,1 \} ^{n} \rightarrow \{ 0,1 \} ^{m}$ be an adaptive $(2,\delta,1/2+\varepsilon)$-RLDC. 
    Then $m = 2^{\Omega_{\delta,\varepsilon}(n)}$.
\end{theorem}
Our results are the first exponential bounds for RLDCs. 
Furthermore, combined with the constructions with nearly linear codeword length for some constant number of queries \cite{Ben-SassonGHSV06, AsadiS21}, our results imply that RLDCs experience a ``phase transition''-type phenomena, where the codeword length drops from being exponential at $q=2$ queries to being only slightly super-linear at $q=c$ queries for some constant $c> 2$. 
In particular, this also implies that there is a query number $q$ where the codeword length drops from being super-polynomial at $q$ to being polynomial at $q+1$. 
Finding this exact threshold query complexity is an intriguing open question.    

\added[id=R3,comment={(5)}]{More concretely, we show that for any $2$-query RDLC with perfect completeness, there exists a subset of input bits $I \subseteq [n]$ of size $\Omega(n)$ such that the code restricted to these input bits is a $2$-query LDC, which then gives us \cref{thm:main} by the results of \cite{KerenidisW04}.
We state (and later prove) this in the following proposition.
\begin{proposition}\label{prop:adaptive-2qRLDC}
    Let $C \colon \{0,1\}^{n} \rightarrow \{0,1\}^{m}$ be an adaptive $(2, \delta, 1/2+\varepsilon)$-RLDC.
    Then, there exists $I \subseteq [n]$ of size $n' = \Omega(n)$ such that decoding the code $C$ restricted to $i \in I$ is a $(2, \delta/2, 1/2+\varepsilon)$-adaptive LDC.
\end{proposition}}

\subsection{Technical Overview}
We give an overview of our main result (\Cref{thm:main}\added{, \cref{prop:adaptive-2qRLDC}}) here.
To simplify the presentation, we assume a non-adaptive decoder in this overview. 
While the exact same arguments do not directly apply to adaptive decoders,\footnote{\label{foot:kt-obs}For standard LDCs, Katz and Trevisan \cite{KatzT00} observed that an adaptive decoder can be converted into a non-adaptive decoder by randomly guessing the output $y_j$ of the first query $j$ to learn the second query $k$. 
Now, we non-adaptively query the received codeword for both $y_j$ and $y_k$. 
If our guess for $y_j$ was correct, then we continue simulating the adaptive decoder. 
Otherwise, we simply guess the output $x_i$. 
If the adaptive decoder succeeds with probability at least $p \geq 1/2 + \varepsilon$, then the non-adaptive decoder succeeds with probability $p' \geq 1/4 + p/2 \geq 1/2 + \varepsilon/2$. 
Unfortunately, this reduction does not preserve perfect completeness, which is required by our proof; \ie, if $p=1$ then $p' = 3/4$.} 
with a bit more care they can be adapted to this setting. 

At a high level, we prove our lower bound by transforming any non-adaptive $2$-query weak Hamming RLDC for messages of length $n$ and $\delta$ fraction of errors into a standard $2$-query Hamming LDC for messages of length $n'=\Omega(n)$, with slightly reduced error tolerance of $\delta/2$. 
Kerenidis and de Wolf \cite{KerenidisW04} proved that any $2$-query Hamming LDC for messages of length $n$ must have codeword length $m = \exp(\Omega(n))$. 
Combining this result with our transformation, it immediately follows that any $2$-query weak Hamming RLDC must also have codeword length $m = \exp(\Omega(n))$. 
While our transformation does not need the third property (success rate) of a strong RLDC, we crucially rely on the property of \emph{perfect completeness} and that the decoder only makes \added{at most }$q=2$ queries.

Let $C\colon\{0,1\}^n \rightarrow \{0,1\}^m$ be a weak $(2,\delta,1/2+\varepsilon)$-RLDC. 
For simplicity (and without loss of generality), let us assume the decoder $\mathsf{Dec}$ works as follows. 
For message $x$ and input $i \in [n]$, the decoder non-adaptively makes 2 random queries $j ,k \in [m]$, and outputs $f_{j,k}^{i}(y_j, y_k) \in \{0,1,\perp\}$, where $y_j, y_k$ are answers to the queries from a received word \replaced{$\mathbf{y} \in \{0,1\}^{m}$}{$y$}, and $f_{j,k}^{i} \colon \{ 0,1 \} ^{2} \rightarrow \{0,1,\perp\}$ is a deterministic function. 
When there is no error, we have $y_j = C(\mathbf{x})_j$ and $y_k = C(\mathbf{x})_k$\added{ for some input $\mathbf{x} \in \{0,1\}^{n}$}. 

We present the main ideas below, and refer the readers to \Cref{sec:2qrldc} for full details.

\paragraph{Fixable Codeword Bits.} 
The starting point of our proof is to take a closer look at the decoding functions $f_{j,k}^{i}$ which contain $\bot$ in their truth tables. 
\added[id=R3]{First, we identify $m$ Boolean functions with the encoder $C$, defined as $C(\mathbf{x}) = (C_{1}(\mathbf{x}), \dotsc, C_{m}(\mathbf{x}))$, where $C_{\ell} \colon \{0,1\}^{n} \rightarrow \{0,1\}$ for each $\ell \in [m]$. 
Then, for any $C_{\ell}$ and $i \in [n]$ we define $C_{\ell}\restriction_{x_{i}=b}$ as the restriction of $C_{\ell}$ to all inputs $\mathbf{x} \in \{0,1\}^{n}$ such that $x_{i} = b$. 
With this in mind, we turn our attention to the functions $f_{j,k}^{i}$.}
\replaced[id=R3]{Consider the truth table (with 4 entries) of $f_{j,k}^{i}$, and suppose there is an entry that contains a $\bot$. 
It turns out that in this case, there exists values $b_{j}, b_{k} \in \{0,1\}$ such that $C_{j}\restriction_{x_{i}=b_{j}}$ and $C_{k}\restriction_{x_{i}=b_{k}}$ are \emph{constant functions}. 
More explicitly, there exists $c_{j}, c_{k} \in \{0,1\}$ such that for any $\mathbf{x}^{(j)}$ and $\mathbf{x}^{(k)}$ satisfying $x^{(j)}_{i} = b_{j}$ and $x^{(k)}_{i} = b_{k}$, we have $C_{j}\restriction_{x_{i}=b_{j}}\equiv c_{j}$ and $C_{k}\restriction_{x_{i}=b_{k}}\equiv c_{k}$.}{It turns out that when $f_{j,k}^{i}$ has at least one $\perp$ entry in the truth table, $C(x)_j$ can be fixed to a constant by setting either $x_i=0$ or $x_i=1$, and \added[id=R2]{the }same for $C(x)_k$.} 

To see this, note that the property of perfect completeness forces \replaced{$f_{j,k}^{i}(y_{j}, y_{k})$}{$f_{j,k}^{i}$} to be $0$ or $1$ whenever $x_i=0$ or $x_i=1$ and there is no error \added{(i.e., $y_{j} = C_{j}(x)$ and $y_{k} = C_{k}(x)$)}. 
\replaced[id=R3]{Now suppose that $C_{j}\restriction_{x_{i}=b_{j}}$ is not a constant function for both $b_{j} = 0$ and $b_{j}=1$. 
This implies that for each $b_{j} \in \{0,1\}$, there exists a pair of inputs $\hat{\mathbf{x}}^{(j,b_{j})} \neq \bar{\mathbf{x}}^{(j,b_{j})}$ satisfying $\hat{x}^{(j,b_{j})}_{i} = \bar{x}^{(j,b_{j})}_{i} = b_{j}$ such that $\hat{\mathbf{y}}^{(j,b_{j})} = C_{j}\restriction_{x_{i} = b_{j}}(\hat{x}^{(j,b_{j})}) \neq C_{j}\restriction_{x_{i} = b_{j}}(\bar{x}^{(j,b_{j})}) = \bar{\mathbf{y}}^{(j,b_{j})}$; in other words, $\hat{\mathbf{y}}^{(j,b_{j})} = 1-\bar{\mathbf{y}}^{(j,b_{j})}$ for any $b_{j} \in \{0,1\}$. 
Perfect completeness of the decoder requires that $f_{j,k}^{i} = b_{j}$ given \emph{either} $\hat{\mathbf{y}}^{(j,b_{j})}$ or $\bar{\mathbf{y}}^{(j,b_{j})}$. 
That is, $f_{j,k}^{i}(0, \cdot) = f_{j,k}^{i}(1, \cdot) = b_{j}$. 
This immediately implies that the truth table of $f_{j,k}^{i}$ must have two output entries of $0$ and $1$, which leaves no space for $\bot$ (see \Cref{clm:bot-fix}).}{Thus, if neither $x_i=0$ nor $x_i=1$ fixes $C(x)_j$, then there must be two entries of $0$ and two entries of $1$ in the truth table of $f_{j,k}^{i}$, which leaves no space for $\perp$ (see ).}
Therefore, when there is at least one $\perp$ entry in the truth table of $f_{j,k}^{i}$, we say that \added{the functions }$C_j$ and $C_k$ are \emph{fixable} by $x_{i}$ \added{(i.e., $\exists b_{j}, b_{k} \in \{0,1\}$ such that $C_{j}\restriction_{x_{i}=b_{j}}$ and $C_{k}\restriction_{x_{i}=b_{k}}$ are both constant functions)}. 

This motivates the definition of the set $S_i$, which contains all indices $j \in [m]$ such that the \replaced{function $C_{j}$ }{codeword bits $C(x)_j$ }is fixable by $x_i$; and the definition of $T_j$, \replaced{which contains all }{the set of all }indices $i \in [n]$ such that \replaced{the function $C_{j}$ }{$C(x)_j$ }is fixable by \deleted{the message bit }$x_i$. 
It is also natural to pay special attention to queries $j,k$ \added{such that at least one of these queries is not} contained in $S_i$, since in this case the function $f_{j,k}^{i}$ never outputs $\perp$.
\added{Moreover, with the above discussion in mind, if the truth table of $f_{j,k}^{i}$ contains a $\bot$, then $j,k \in S_{i}$ (\Cref{clm:bot-fix}).}

\paragraph{The Query Structure.} 
In general, a query set $\{ j,k \} $ falls into one of the following three cases: (1) both $j,k$ lie inside $S_i$; (2) both $j,k$ lie outside $S_i$; (3) one of them lies inside $S_i$ and the other lies outside $S_i$. 
\replaced[id=R1]{It turns out that if case (3) occurs with $j \in S_i$ and $k \not \in S_i$, then not only will the decoder never output $\bot$ (\cref{clm:bot-fix}), but also the output of $f_{j,k}^{i}$ \emph{does not depend} on the queried symbol $y_{k}$; i.e., the truth table of $f_{j,k}^{i}$ is completely determined by $y_{j}$ only (see \Cref{clm:fixable-or-useless}).}{It turns out that case (3) essentially never occurs for a decoder with perfect completeness. The reason is that when, say, $j \in S_i$ and $k \notin S_i$, one can effectively pin down every entry in the truth table of $f_{j,k}^{i}$ by using the perfect completeness property, and observe that the output of $f_{j,k}^{i}$ does not depend on $y_k$ at all} 
Thus, in this case we can equivalently view the decoder as only querying $y_j$ where $j \in S_i$, which leads us back to case (1) \added[id=R1]{(e.g., by taking $k = j$ or $k$ to be any other element of $S_{i}$ since $f_{j,k}^{i}$ only requires $y_{j}$)}. 
In what follows, we denote by $E_1$ the event that case (1) occurs, and by $E_2$ the event that case (2) occurs.

\paragraph{The Transformation by Polarizing Conditional Success Probabilities.} 
We now give a high level description of our transformation from a weak RLDC to a standard LDC. 
Let \added{$\mathbf{x} \in \{0,1\}^{n}$ and let} $\mathbf{y} \in \{0,1\}^{m}$ be a string which contains at most $\delta m/2$ errors from the codeword $C(\mathbf{x})$. 
We have established that the success probability of the \deleted{weak} RLDC decoder on $\mathbf{y}$ is an average of two conditional probabilities:
\begin{align*}
	\Pr[\mathsf{Dec}(i,\mathbf{y}) \in \{x_i, \perp\}] = p_1 \cdot \Pr[\mathsf{Dec}(i,\mathbf{y}) \in \{x_i, \perp\} \mid E_1] + p_2 \cdot \Pr[\mathsf{Dec}(i,\mathbf{y}) \in \{x_i, \perp\} \mid E_2],
\end{align*}
where $p_1 = \Pr[E_1]$ and $p_2 = \Pr[E_2]$. 
Let us assume for the moment that $S_i$ has a small size, \eg, $|S_i|\le \delta m/2$. 
The idea in this step is to introduce additional errors to the $S_i$-portion of $\mathbf{y}$, in a way that drops the conditional success probability $\Pr[\mathsf{Dec}(i,\mathbf{y}) \in \{x_i, \perp\} \mid E_1]$ to 0 (see \Cref{lem:conditional-zero}). 

In particular, we modify the bits $y_{j}$ for $j \in S_i$ to make it consistent with the encoding of any message $\hat{\mathbf{x}}$ with $\hat{x}_i=1-x_i$. 
Perfect completeness thus forces the decoder to output $1-x_i$ conditioned on $E_1$. 
\added{In particular, this modification forces the probability $\Pr[\mathsf{Dec}(i,\mathbf{y}) \in \{ x_{i}, \bot \} ~|~ E_{1}]$ to become zero.}
Note that we have introduced at most $\delta m/2 + |S_i| \le \delta m$ errors in total, meaning that the decoder \replaced{still has }{should still have }an overall success probability of $1/2+\varepsilon$. 

\replaced{Now, }{Furthermore, now }the conditional probability $\Pr[\mathsf{Dec}(i,\mathbf{y}) \in \{x_i, \perp\} \mid E_2]$ takes all \replaced[id=R2]{the credit }{credits }for the overall success probability\added{, and we have just established that this probability is $\geq 1/2+\varepsilon$}. 
\replaced[id=R1]{Recalling the fact that $E_{2}$ is the event that $j,k \notin S_{i}$ and that this implies that $f_{j,k}^{i}$ contains no $\bot$ in its truth table, we observe that $\mathsf{Dec}$ never outputs $\bot$ conditioned on event $E_{2}$.}
{Combined with the observation that $\mathsf{Dec}$ never outputs $\perp$ given $E_2$,} 
This suggests the following natural way to decode $x_i$ in the sense of a \emph{standard} LDC: sample queries $j, k$ according to the conditional probability given $E_2$ (\ie, both $j,k$ lie outside $S_i$) and output $f_{j,k}^{i}(y_j, y_k)$. 
This gives a decoding algorithm for standard LDC, with success probability $1/2+\varepsilon$ and error tolerance $\delta m/2$ (see \Cref{lem:LDC-reduction}), modulo the assumption that $|S_i|\le \delta m/2$.\footnote{\added{Actually, by our previous discussion on the query structure, it suffices for at least one of $j,k$ to lie outside $S_{i}$, since $j \in S_{i}$ and $k \notin S_{i}$ implies that we can ignore query $k$ and treat it as if it were in the set $S_{i}$.}}

\paragraph{Upper Bounding \texorpdfstring{$|S_i|$}{|S\_i|}.} 
The final piece in our transformation from \replaced{a }{weak }RLDC to \added{a }standard LDC is to address the assumption that $|S_i|\le \delta m/2$. 
This turns out to be not true in general, but it would still suffice to prove that $|S_i| \leq \delta m/2$ for $n' = \Omega(n)$ of the message bits $i$. 
If we could show that $|T_j|$ is small for most $j \in [m]$, then a double counting argument shows that $|S_i|$ is small for most $i \in [n]$. 
Unfortunately, if we had  $C_j(\mathbf{x}) = \bigwedge_{i=1}^n x_i$ for $m/2$ of the codeword bits $j$, then we also have $|T_j| = n$ for $m/2$ codeword bits and $|S_i| \geq m/2 \geq \delta m/2$ for all message bits $i \in [n]$. 
We address this challenge by first arguing that any weak RLDC for $n$-bit messages can be transformed into another weak RLDC for $\Omega(n)$-bit messages which satisfy $|T_j| \le 3\ln(8/\delta) $ for all but $\delta m/4$ codeword bits. 

The transformation works by fixing some of the message bits and then eliminating codeword bits that are fixed to constants. 
Intuitively, if some $C_j$ is fixable by many message bits, it will have very low entropy (e.g., $C_j$ is the AND of many message bits) and hence contain very little information and can (likely) be eliminated\added{ (i.e., querying these message bits does not help the decoder, so we can ignore them)}. 
We make this intuition rigorous through the idea of random restriction: for each $i \in [n]$, we fix $x_i=0$, $x_i=1$, or leave $x_i$ free, each with probability $1/3$. 
Then, the probability that $C_j$ is not fixed to a constant\added{ by this random restriction} is at most $(1-1/3)^{|T_j|}\le \delta/8$, provided that $|T_j| \ge 3\ln(8/\delta)$. 

After eliminating codeword bits that are fixed to constants, we show that with probability at least $1/2$, at most $\delta m/4$ codeword bits $C_j(\mathbf{x})$ with $|T_j| \ge 3\ln(8/\delta)$ survived.\footnote{\label{foot:tj-restriction}We are oversimplifying a bit for ease of presentation. In particular, the random restriction process may cause a codeword bit $C_j(\mathbf{x})$ to be fixable by a new message bit $x_i$ that did not belong to $T_j$ before the restriction---we thank an anonymous reviewer for pointing this out to us. Nevertheless, for our purpose it is sufficient to eliminate codeword bits that initially have a large $|T_j|$. See the formal proof for more details.} 
Note that with high probability \added{(i.e., except with $\exp(-\Omega(n))$ probability)}, the random restriction leaves at least $n/6$ message bits free, implying that there must exist a restriction which leaves at least $n/6$ message bits free.
This ensures that $|T_j| \ge 3\ln(8/\delta)$ for at most $\delta m/4$ of the remaining codeword bits $C_j(\mathbf{x})$. 
We can now apply the double counting argument to conclude that $|S_i| \le \delta m/2$ for $\Omega(n)$ message bits, completing the transformation.

\paragraph{Adaptive Decoders.} For adaptive decoders, we are going to follow essentially the same proof strategy. 
The new idea and main difference is that we focus on the first query made by the decoder, which is always non-adaptive. 
We manage to show that the first query determines a similar query structure, which is the key to the transformation to a standard LDC. 
\added{At a high level, we show that so long as the first query lies outside the set $S_{i}$, the decoder never outputs $\bot$. Moreover, we show that if the first query does lie within $S_{i}$, then if the second query lies outside $S_{i}$, the decoder essentially only needs the query $y_{j}$ to make its decision, while if the second query lies within $S_{i}$, then we can construct an input to the decoder which flips its decision with probability one. 
We then use these facts to construct a standard adaptive LDC in a manner similar to the non-adaptive case, though this reduction does not preserve perfect completeness of the RLDC decoder.
However, this reduction still works since the 2-query exponential lower bound of Kerenidis and de Wolfe \cite{KerenidisW04} still holds for LDCs with imperfect completeness error.}
More details can be found in \Cref{subsec:adaptive-2qRLDC}.

\subsection{Open Questions}\label{sec:open-questions}
\paragraph{Exact ``Phase-transition'' Threshold.}
\replaced[id=R1]{%
Our results show that RLDCs and standard LDCs are equivalent when $q = 2$; in particular, $2$-query Hamming RLDCs require super-polynomial codeword length.
Thus, our results imply that there exists a $q = O(1)$ such that every $q$-query RLDC requires super-polynomial codeword length, while there exists a $(q+1)$-query RLDC of polynomial codeword length.
}

Our results show that for Hamming RLDCs there exists a $q$ such that every $q$-query RLDC requires super-polynomial codeword length, while there exists a $(q+1)$-query RLDC of polynomial codeword length.
Finding the precise $q$ remains an intriguing open question. 
Further, a more refined understanding of codeword length for RLDCs making $3, 4, 5$ queries is another important question, which has \replaced[id=R2]{led }{lead }to much progress in the understanding of other LDC variants.


\paragraph{Lower Bounds for Hamming (R)LDCs.}
Our $2$-query lower bound for Hamming RLDCs crucially uses the perfect completeness property of the decoder. 
An immediate question is whether the bound still holds if we allow the decoder to have imperfect completeness. 
We also note that the argument in our exponential lower bounds for $2$-query Hamming RLDCs fail\added[id=R3]{s} to hold for alphabets other than the binary alphabet, and we leave the extension to larger alphabet sizes as an open problem. 
Another related question is to understand if one can leverage perfect completeness and/or random restrictions to obtain improved lower bounds for $q\geq 3$-query standard Hamming LDCs. 
Perfect completeness has been explicitly used before to show exponential lower bounds for $2$-query LCCs \cite{bhattacharyya2017lower}, so such techniques seem plausible.

\subsection{Other Related Work}

\paragraph{Phase-transition Threshold.}
\added[id=R1]{%
While progress towards resolving the phase-transition threshold question has not been made directly, Grigorescu et al. \cite{GKMM25} make some progress on a related question: the equivalence of RLDCs and LDCs.
In particular, they study finding the smallest constant $r$ such that for any $q < r$, every $q$-query RLDC is a $q$-query LDC, and there is an $r$-query RLDC that is not an $r$-query LDC.
They study this question in light of our results: they observe our results actually show that any $2$-query RLDC is a $2$-query LDC, and thus are equivalent.
However, there exists a larger (but constant) $q$ such that (strong) $q$-query RLDCs exist with $m = n^{1+\Theta(1/q)}$ \cite{Ben-SassonGHSV06,GurRR20,ChiesaGS20,Goldreich25,GoldbergGS26}, while the best known $q$-query LDC constructions achieve $m = 2^{n^{o(1)}}$ \cite{Yekhanin08,Efremenko12}, suggesting there exists such a constant $r$ for which RLDCs and LDCs are no longer equivalent.}

\added[id=R1]{%
    Grigorescu et al. \cite{GKMM25} prove that this $r$ lies in the range $[4,15)$ for \emph{linear} RLDCs over the binary alphabet (i.e., (R)LDCs where the encoding algorithm is a linear map) by showing that any $3$-query linear RLDC is equivalent to a linear $3$-query LDC, while also giving a construction of a $15$-query linear RLDC which is not an $O(\log(n))$-query linear LDC.
    Moreover, these results also hold for binary linear RLCCs (with a different range), where they show that linear 3-query RLCCs are equivalent to linear 3-query LCCs, while $41$-query linear RLCCs are not $O(\log(n))$-query linear LCCs.
}

\added{%
    In a recent follow-up work, Gur et al. \cite{GurMWZ26} showed, in fact, that $r=2$, by giving a construction of a non-linear $3$-query RLDC with constant alphabet size (but non-binary) and length $m = \widetilde{O}(n^{2})$ for $n$-bit messages.
    This implies a strict separating between $3$-query RLDCs and $3$-query LDCs, as Alrabiah et al. \cite{AlrabiahGKM} give a lower bound of $m = \widetilde{\Omega}(n^{3})$ for $3$-query LDCs (with the same parameters).
    Moreover, this also gives progress on the phase-transition question we propose in \cref{sec:open-questions}, showing the threshold must be at least $4$. 
    \cite{GurMWZ26} suggest that, due to the equivalence of linear 3-query RLDCs and LDCs of \cite{GKMM25}, a non-binary alphabet is ``likely a prerequisite for achieving such polynomial separation.''
}

\added{%
    Even more recently, Cheng, Li, and Mao \cite{CLM26} take a different approach to this question and show that every $q$-query RLDC/RLCC with sufficiently small relaxed decoding error $1-\alpha$ is a $q$-query LDC/LCC with comparable parameters; in particular $1-\alpha < 2^{-q}$.
    This work removed the linear encoding requirement of \cite{GKMM25} and improves the parameters of the obtained LDC/LCC.
    Moreover, they obtain comparable results even when the RLDC/RLCC has imperfect completeness.
}

\deleted[id=R1]{%
Some progress has been made towards resolving the phase-transition threshold question.
For \emph{linear} Hamming RLDCs (i.e., (R)LDCs where encoding algorithm is a linear map)\deleted[id=R2]{(i.e., (R)LDCs where encoding algorithm is a linear map)}, \cite{GKMM25} show that this $q$ lies in the range $[4,15)$.
They show that for $q = 3$, linear $3$-query RLDCs are equivalent to linear $3$-query LDCs, while also giving a construction of a $q=15$ query linear RLDC which \emph{is not} a constant-query (more specifically, \replaced[id=R2]{$\Omega(\log(n))$}{$O(\log(k))$}-query for message length\replaced[id=R2]{ $n$}{ $k$}) LDC, separating constant-query linear RLDCs and LDCs.
Moreover, these results also hold for RLCCs, where linear $3$-query RLCCs are equivalent to linear $3$-query LCCs, and linear $41$-query RLCCs are not constant-query (\replaced[id=R2]{$O(\log(n))$}{$O(\log(k))$}-query) linear LCCs.
}

\paragraph{Lower Bounds.}
We discussed the various best-known RLDC lower bounds for constant $q$-queries previously in the introduction (e.g., \cite{GurL21,DallAgnolGL23,GoldbergGS26}), and remark that these lower bounds translate to $m=\widetilde{\Omega}(\sqrt{\log(n)})$ for poly-logarithmic query complexity $q = \polylog(n)$.
\deleted{{The best-known RLDC} lower bound\deleted[id=R2]{s} for constant $q\geq3$ queries is due to Gur and Lachish \cite{GurL21}, which states that every $q$-query RLDC satisfies \replaced[id=R2]{$m=\Omega(n^{1+\alpha}$ }{$n = \Omega(k^{1+\alpha})$ }for constant $\alpha$ (which depends on $q$ and the decoding radius $\delta$), which matches the upper bound of Ben-Sasson et al. \cite{Ben-SassonGHSV06}.
For poly-logarithmic query complexity, this lower bound translates to \replaced[id=R2]{$\widetilde{\Omega}(\sqrt{\log(n)})$}{$\widetilde{O}(\sqrt{\log(n)})$}.}
Recently, Kumar and Mon \cite{KumarM24} showed \replaced[id=R2]{an upper bound of }{a nearly matching upper bound }$O(\log^{69}(n))$ queries for constant-rate RLCCs (and thus RLDCs), showing that this lower bound is (asymptotically) tight (with room for concrete improvement).
Their construction starts with a RLCC with very small block length (i.e., poly-logarithmic in $n$), then continuously increasing the block length to construct a larger RLCC by utilizing suitable Locally Testable codes.
The composition is done in such a way that the query complexity only increases additively rather than exponentially (as with prior work \cite{GurRR20}) in the number of iterations to increase the block length.
Subsequently, this result was improved by Cohen and Yankovitz \cite{CohenY24}, giving a $\log^{2+o(1)}(n)$ query upper bound, by observing that the LTCs used by Kumar and Mon in their construction are ``overkill'' and can be replaced with expander codes with some suitable properties\added[id=R1]{ (i.e., local testability under a promise that the corruption rate is sufficiently small, which is satisfied by expander codes)}.

\added{%
Cheng, Li, and Mao \cite{CLM26} also give new lower bounds for any constant $q$-query RLDC/RLCC with constant error-tolerance.
By combining their results with the results of \cite{JanzerM25,AlrabiahGKM}, they show that any constant $q$-query RLDC/RLCC with $1-\alpha < 2^{-q/2}$ and constant error-tolerance satisfies $m = \widetilde{\Omega}(n^{q/(q-2)})$.
}

\comment[id=R2]{Add discussion for \cite{CLM26,GurMWZ26}.}
\comment[id=R1]{R1 also adds this comment.}

\section{Preliminaries}\label{sec:prelims}
For natural number $n \in \N$, we let $[n] := \{1,2,\dotsc, n\}$.
For a string $\mathbf{x} \in \{ 0,1 \} ^*$ of finite length, we let $|\mathbf{x}|$ denote the length of $\mathbf{x}$.
For $i \in [|\mathbf{x}|]$, we let $x_{i}$ denote the $i$-th bit of $\mathbf{x}$. 
Furthermore, for $I \subseteq [|\mathbf{x}|]$, we let $\mathbf{x}_{I}$ denote the subsequence $(x_{i_{1}}, x_{i_{2}}, \dotsc, x_{i_{\ell}})$, where $i_j \in I$ and $\ell = |I|$, \added{and for index $i \in [n]$, we let $\mathbf{x}_{-i} \coloneqq \mathbf{x}_{[n]\setminus \{ i \} }$.}
For two strings $\mathbf{x}, \mathbf{y} \in \{ 0,1 \} ^n$ of length $n$, we let $\mathsf{HAM}(\mathbf{x},\mathbf{y})$ denote the \emph{Hamming Distance} between $\mathbf{x}$ and $\mathbf{y}$; \ie, $\mathsf{HAM}(\mathbf{x},\mathbf{y}) \coloneqq |\{ i \in [n] : x_i \neq y_i \}|$.
We often discuss the \emph{relative Hamming Distance} between $\mathbf{x}$ and $\mathbf{y}$, which is simply the Hamming Distance normalized by $n$, \ie, $\mathsf{HAM}(\mathbf{x},\mathbf{y})/n$; for ease of presentation, we let $\mathsf{ham}(\mathbf{x},\mathbf{y}) := \mathsf{HAM}(\mathbf{x},\mathbf{y})/n$.
Finally, the \emph{Hamming weight} of a string $\mathbf{x}$ is the number of non-zero entries of $\mathbf{x}$, which we denote as $\mathsf{wt}(\mathbf{x}) := |\{i \in [|\mathbf{x}|] : x_i \neq 0\}|$; equivalently, $\mathsf{wt}(\mathbf{x}) = \mathsf{HAM}(\mathbf{x}, \mathbf{0}^{|\mathbf{x}|})$, where $\mathbf{0}^{i}$ is the all-zero string of length $i$. 

\added[id=R1]{%
Given a Boolean function $f \colon \{0,1\}^{n} \rightarrow \{0,1\}$, and $\sigma \in \{0,1\}$, we write $f\restriction_{x_i = \sigma}$ to denote the restriction of $f$ to the domain $\{\mathbf{x}\in \{0,1\}^n : x_i = \sigma\}$. 
For a sequence of restrictions, we simply write $f\restriction_{(x_{j_1}, \dots, x_{j_k})=(\sigma_1,\dotsc,\sigma_k)}$, or $f_{J|\sigma}$ where $J=[n]\setminus\{j_1,\dotsc,j_k\}$ and $\boldsymbol{\sigma}=(\sigma_1,\dotsc,\sigma_k)$. 
Note that $f_{J|\boldsymbol{\sigma}}$ is a Boolean function over the domain $\{ 0,1 \} ^{|J|}$. 
}

For completeness, we recall the definition of a classical locally decodable code, or just a \emph{locally decodable code}.
\begin{definition}[Locally Decodable Codes]\label{def:LDC}
    A $(q, \delta, \alpha)$-Locally Decodable Code $C \colon \Sigma^n \rightarrow \Sigma^m$ is a code for which there exists a randomized decoder that makes at most $q$ queries to the received word $\mathbf{y}$ and satisfies the following property: for every $i \in [n]$, if $\mathbf{y}$ is such that $\mathsf{ham}(\mathbf{y}, C(\mathbf{x})) \leq \delta$ for some unique $C(\mathbf{x})$, then the decoder, on input $i$, outputs $x_i$ with probability $\geq \alpha$.
    Here, the randomness is taken over the random coins of the decoder, and $\mathsf{ham}$ is the normalized Hamming distance.
    \added{When $\Sigma=\{0,1\}$, we say that $C$ is a \emph{binary} LDC.}
\end{definition}

We recall the general $2$-query Hamming LDC lower bound  \cite{KerenidisW04, Ben-AroyaRW08}, \added{which applies to both adaptive and non-adaptive decoders.}
\begin{theorem}[\cite{KerenidisW04, Ben-AroyaRW08}]\label{thm:two-query-lb}
    For constants $\delta, \varepsilon \in (0,1/2)$ there exists a constant $c = c(\delta, \varepsilon) \in (0,1)$ such that if $C \colon \{ 0,1 \} ^n \rightarrow \{ 0,1 \} ^m$ is a $(2, \delta, 1/2+\varepsilon)$ Hamming LDC then $m \geq 2^{cn-1}$.
\end{theorem}

\section{Lower Bounds for 2-Query Hamming RLDCs}\label{sec:2qrldc}
We prove \Cref{thm:main} in this section. 
As a reminder, a weak $(q,\delta,\alpha)$-RLDC satisfies the first two conditions in \Cref{def:strongRLDC}, and non-adaptive means the decoder makes queries according to a distribution which is independent of the received string $y$. 
Here, we are interested in the case $q=2$ and $\alpha=1/2+\eps$.

To avoid overloading first-time readers with heavy notation, we first present a proof of the lower bound for \emph{non-adaptive} decoders, \ie, decoders with a query distribution independent of the received string. 
This proof will be easier to follow, while the crucial ideas behind it remain the same. 
The proof for the most general case is presented in the last subsection (\Cref{subsec:adaptive-2qRLDC}), with an emphasis on the nuances in dealing with adaptivity.

\subsection{Warm-up: The Lower Bound for Non-adaptive Decoders}
In the following, we fix a relaxed decoder $\mathsf{Dec}$ for code $C$. 
In this subsection, we assume that $\mathsf{Dec}$ is non-adaptive, and that it has the first two properties specified in \Cref{def:strongRLDC}. 
To avoid technical details, we also assume $\mathsf{Dec}$ always makes exactly 2 queries (otherwise add dummy queries to make the query count exactly 2).

Given an index $i \in [n]$ and queries $j,k$ made by $\mathsf{Dec}(i,\cdot)$, in the most general setting\added[id=R1]{,} the output could be a random variable which depends on $i$ and $y_j$, $y_k$, where $y_j$, $y_k$ are the answers to queries $j$, $k$, respectively. 
An equivalent view is that the decoder picks a random function $f$ according to some distribution and outputs $f(y_j, y_k)$. 
Let $\mathtt{DF}^i_{j,k}$ be the set of all decoding functions $f \colon \{0,1\}^2 \rightarrow \{0,1,\perp\}$ which are selected by $\mathsf{Dec}(i,\cdot)$ with non-zero probability when querying $j,k$. 
We partition the queries into the following two sets
\begin{align*}
    F_{i}^{0} &:= \left\{ \{j,k\} \subseteq [m] : \forall f \in \mathtt{DF}^i_{j,k} \text{ the truth table of $f$ contains no ``$\perp$''}\right\},\\
    F_{i}^{\ge 1} &:= \left\{ \{j,k\} \subseteq [m] : \exists f \in \mathtt{DF}^i_{j,k} \text{ the truth table of $f$ contains at least 1 ``$\perp$''}\right\}.
\end{align*}


We identify the encoding function of $C$ as a collection of $m$ Boolean functions
\begin{align*} 
    \mathcal{C} := \left\{C_1, \dotsc, C_m : \forall j \in [m], C_j \colon \{0,1\}^n \rightarrow \{0,1\}\right\}.
\end{align*}
Namely, $C(\mathbf{x})=(C_1(\mathbf{x}),\dotsc, C_m(\mathbf{x}))$ for all $\mathbf{x} \in \{ 0,1 \} ^n$.
\added{We now define the notion of \emph{fixable codeword bits}.
\begin{definition}[Fixable Codeword Bits]\label{def:fixable}
    Let $\mathcal{C} \coloneqq \left\{C_1, \dotsc, C_m : \forall j \in [m], C_j \colon \{0,1\}^n \rightarrow \{0,1\}\right\}$ be the encoding function of $C$.
    For $j \in [m]$, we say that $C_{j}$ is \emph{fixable by $x_{i}$} if there exists $b \in \{0,1\}$ and $c_{j} \in \{0,1\}$ such that for all $\mathbf{x} \in \{0,1\}^{n}$ such that $x_{i} =b$, we have $C_{j}(\mathbf{x}) = c_{j}$.
    When $C_{j}$ is fixable by $x_{i}=b$ with value $c_{j}$, we write $C_{j}\restriction_{x_{i}=b}\equiv c_{j}$.
\end{definition}}
\added{In other words,} for $j \in [m]$, $C_j$ is {fixable} by $x_i$ if at least one of the restrictions $C_j\restriction_{x_i=b}$ for $b \in \{0,1\}$ is a constant function. 
\added{Now, define the following two sets with respect to \cref{def:fixable}:}
\begin{align}
    S_i &\coloneqq \{j \in [m] \colon C_j\text{ is fixable by }x_i\}\label{eq:Si}\\
    T_j &\coloneqq \{i \in [n] \colon C_j\text{ is fixable by }x_i\}.\label{eq:Tj}
\end{align}

\replaced{Now, given \cref{def:fixable,eq:Si,eq:Tj}, we first show that, due to the perfect completeness property of $C$, the queries of $\mathsf{Dec}$ must have certain structures.}{We first show that the queries of $C$ must have certain structures.} 
The following claim characterizes the queries in $F_i^{\geq 1}$.
\begin{claim} \label{clm:bot-fix}
    \replaced[id=R1]{Let $j,k \in [m]$ and suppose $j\notin S_{i}$ or $k \notin S_{i}$.
    Then, $\{ j,k \}  \notin F_{i}^{\geq 1}$; in particular, $\{ j,k \} \in F_{i}^{0}$.}{%
    Suppose $\{j,k\} \in F_{i}^{\ge 1}$. 
    Then we must have $j,k \in S_i$.}
\end{claim}
\begin{proof}
    \replaced[id=R1]{%
        Let $j,k \in [m]$ and suppose without loss of generality that $j \notin S_{i}$.
        }{%
        Let $\{j, k\} \in F_{i}^{\ge 1}$. 
        Suppose for the sake of contradiction that $j \notin S_i$.}
    This implies there are partial assignments $\boldsymbol{\sigma}_{00}, \boldsymbol{\sigma}_{01}, \boldsymbol{\sigma}_{10}, \boldsymbol{\sigma}_{11} \in \{0,1\}^{n-1}$ such that 
	\begin{gather*}
        C_j\left(\mathbf{x}_{-i} = \boldsymbol{\sigma}_{00}, x_i = 0\right) = 0, \quad C_j\left(\mathbf{x}_{-i} = \boldsymbol{\sigma}_{01}, x_i = 1\right) = 0, \\
        C_j\left(\mathbf{x}_{-i} = \boldsymbol{\sigma}_{10}, x_i = 0\right) = 1, \quad C_j\left(\mathbf{x}_{-i} = \boldsymbol{\sigma}_{11}, x_i = 1\right) = 1.
	\end{gather*}
    \deleted{where $\mathbf{x}_{-i}\coloneqq(x_t : t \in [n]\setminus\{i\})$.}
	
	Let $C_{00}, C_{01}, C_{10}, C_{11}$ be encodings of the corresponding assignments mentioned above. 
    Since the relaxed decoder has perfect completeness, when $\mathsf{Dec}(i,\cdot)$ is given access to $C_{00}$ or $C_{10}$, it must output $x_i=0$. 
    Note that the $j$-th bit is different in $C_{00}$ and $C_{10}$. 
    Similarly, when $\mathsf{Dec}(i,\cdot)$ is given access to $C_{01}$ or $C_{11}$, it must output $x_i = 1$; again, the $j$-th bit is different in $C_{01}$ and $C_{11}$. 
    However, this already takes up 4 entries in the truth table of any decoding function $f \in \mathtt{DF}_{j,k}^i$, leaving no space for any ``$\perp$'' entry. 
    \replaced{Therefore, $\{ j,k \} \in F_{i}^{0}$.}{This contradicts with the assumption $\{j,k\} \in F_{i}^{\ge 1}$.} 
\end{proof}

\replaced{Another way to stated \Cref{clm:bot-fix} is that on input $i$, $\mathsf{Dec}(i,\cdot)$ never outputs $\bot$ conditioned on the event that some query $j$ or $k$ is not contained in $S_{i}$.}{%
Here is another way to view which will be useful later: if $\{j, k\}$ is a query set such that $j \notin S_i$ (or $k \notin S_i$), then $\{j, k\} \in F_{i}^{0}$. 
In other words, conditioned on the event that some query is not contained in $S_i$, the decoder never outputs $\perp$.}

The following claim characterizes the queries in $F_i^{0}$.
\begin{claim}\label{clm:fixable-or-useless}
    Suppose $\{j,k\} \in F_{i}^{0}$ and $j \in S_i$. 
    Then one of the following three cases occur: (1) $k \in S_i$, (2) $C_j=x_i$, or (3) $C_j=\neg x_i$.
\end{claim}
\begin{proof}
    Since $j \in S_i$, we may, without loss of generality, assume that $C_j\restriction_{x_i=0}$ is a constant function. 
    Let us further assume it is the constant-zero function. 
    The proofs for the other cases are going to be similar. 
    
    Denote by $f(y_j, y_k)$ the function returned by $\mathsf{Dec}(i,\cdot)$ conditioned on reading $\{j,k\}$. 
    Any function $f \in \mathtt{DF}_{j,k}^i$ takes values in $\{0,1\}$ since $\{j,k\} \in F_i^{0}$. 
    Suppose case (1) does not occur, meaning that $C_k\restriction_{x_i=b}$ is not a constant function \added{for any $b \in \{0,1\}$}. 
    \added{Consider $b = 0$ ($b = 1$ is analogous).}
    Then, there must be partial assignments $\boldsymbol{\sigma}_{00}, \boldsymbol{\sigma}_{10} \in \{0,1\}^{n-1}$ such that
    \begin{align*}
        C_k(\mathbf{x}_{-i}=\boldsymbol{\sigma}_{00}, x_{i}=0) = 0, \quad C_k(\mathbf{x}_{-i}=\boldsymbol{\sigma}_{10}, x_{i}=0) = 1.
    \end{align*}
    Let $C_{00}$ and $C_{10}$ be the encodings of the corresponding assignments mentioned above. 
    Due to perfect completeness of $\mathsf{Dec}$, it must always output $x_i = 0$ when given access to $C_{00}$ or $C_{10}$. 
    \replaced{Since we have assumed that $C_{j}\restriction_{x_{i}=0}$ is the zero function, this implies $f(0,0) = f(0,1) = 0$.}
    {That means $f(0,0)=f(0,1)=0$}
    
    Now we claim that $C_j\restriction_{x_i=1}$ must be the constant-one function. 
    If not, then there are partial assignments $\boldsymbol{\sigma}_{01}, \boldsymbol{\sigma}_{11} \in \{ 0,1 \} ^{n-1}$ such that
    \begin{align*}
        C_j(\mathbf{x}_{-i}={\boldsymbol{\sigma}_{01}}, x_{i}=1) = 0, \quad C_{j}(\mathbf{x}_{-i} = \boldsymbol{\sigma}_{11}, x_{i}=1) = 1.
    \end{align*}
    \added{Now, suppose under these assignments that $C_{k}$ takes the following values:}
    \begin{align*}
        C_k(\mathbf{x}_{-i}={\boldsymbol{\sigma}_{01}}, x_{i}=1) = b, \quad C_{k}(\mathbf{x}_{-i} = \boldsymbol{\sigma}_{11}, x_{i}=1) = b'.
    \end{align*}
    Let $C_{01}, C_{11}$ be the encodings of these assignments. 
    On the one hand, due to perfect completeness $\mathsf{Dec}(i,\cdot)$ should always output $x_i=1$ when given access to $C_{01}$ or $C_{11}$. 
    \added{In particular, $f((C_{01})_{j}, (C_{01})_{k}) = f(0,b) = 1$ must hold.}
    On the other hand, by our previous argument, $\mathsf{Dec}(i,\cdot)$ outputs $f(0,0) = f(0,1)=0$.
    This contradiction shows that $C_j\restriction_{x_i=1}$ must be the constant-one function. 
    Therefore, $C_j=x_i$, i.e., case (2) occurs.
    
    Similarly, when $C_j\restriction_{x_i=0}$ is the constant-one function, we can deduce that $C_j=\neg x_i$, \ie, case (3) occurs.
\end{proof}

\begin{remark}\label{rem:clm12}
    There are three possible scenarios for any query set $\{ j,k \} $ made by the non-adaptive decoder $\mathsf{Dec}$ on input $i$: (1) $j,k \in S_{i}$, (2) $j,k \notin S_{i}$, or (3) $j\in S_{i}$ and $k \notin S_{i}$ (or the reverse).
    In case (1), the output of the decoder depends exclusively on the bits in $S_{i}$; that is, for any string $\mathbf{y} \in \{0,1\}^{m}$, the output of $\mathsf{Dec}(i, \mathbf{y})$ depends only on $\mathbf{y}_{S_{i}}$.
    For cases (2) and (3), \Cref{clm:bot-fix} shows that the decoder never outputs $\bot$.
    Moreover, for case (3), \Cref{clm:fixable-or-useless} shows something stronger: $k$ is a dummy query not used for decoding since $C_{j} = x_{i}$ or $C_{j} = \neg x_{i}$.
    In other words, in case (3), the output of the decoder again only depends on bits in $S_{i}$, which reduces back to case (1) (we can simply ignore the query $k$ or define it as $k=j$).
%
\end{remark}

Let $J \subseteq [n]$ and $\boldsymbol{\rho} \in \{0,1\}^{\overline{J}}$. 
A code $C \colon \{0,1\}^n \rightarrow \{0,1\}^m$ restricted to $\mathbf{x}_{\overline{J}} = \boldsymbol{\rho}$, denoted by $C_{J|\boldsymbol{\rho}}$, is specified by the following collection of Boolean functions
\begin{equation}
    \mathcal{C}_{J|\boldsymbol{\rho}} := \{C_j\restriction_{\mathbf{x}_{\overline{J}}=\boldsymbol{\rho}} : j \in [m], C_j\restriction_{\mathbf{x}_{\overline{J}}=\boldsymbol{\rho}}\text{ is not a constant function}\}.\label{eq:restricted-code}
\end{equation} 
Namely, we restrict each function $C_j$ in $\mathcal{C}$ to $\mathbf{x}_{\overline{J}}=\boldsymbol{\rho}$, and eliminate those that have become constant functions. 
$C_{J|\boldsymbol{\rho}}$ encodes $n'$-bit messages into $m'$-bit codewords, where $n'=|J|$ and $m' = |\mathcal{C}_{J|\boldsymbol{\rho}}| \le m$. 


\begin{remark}\label{rem:restricted-decoder}
We note that the local decoder $\mathsf{Dec}$ for $C$ can be transformed into a local decoder $\mathsf{Dec}'$ for $C_{J|\boldsymbol{\rho}}$, while preserving all parameters, as follows.
Let $\mathcal{L} \subseteq [m]$ denote all indices $\ell$ such that $C_{\ell}\restriction_{\mathbf{x}_{\overline{J}}=\boldsymbol{\rho}}$ is a constant function.
First, $\mathsf{Dec}'$ has the set $J$, assignment $\boldsymbol{\rho}$, and pairs $(\ell,c_{\ell}) \in \mathcal{L} \times \{0,1\}$ hardcoded, where $C_{\ell}\restriction_{\mathbf{x}_{\overline{J}}=\boldsymbol{\rho}} = c_{\ell}$ for each $\ell \in \mathcal{L}$.
For convenience, also assume that $J$ is ordered as $\{ j_{1}, \dotsc, j_{n'} \} $, and that $\overline{\cL}$ is ordered as $\{ \ell_{1}, \dotsc, \ell_{m'}  \} $.
Now, on input $i' \in [n']$ and given received word $\mathbf{y}' \in \{0,1\}^{m'}$, $\mathsf{Dec}(i',\cdot)$ operates as follows.
\begin{itemize}
    \item Using the set $J$, compute the ``true'' index $i \in [n]$ such that $j_{i'} = i$.
    \item Obtain queries $j,k \in [m]$ from $\Dec(i,\cdot)$.
    \item If there is a tuple of the form $(j, c_{j})$ or $(k, c_{k})$, return $c_{j}$ or $c_{k}$ to $\mathsf{Dec}(i,\cdot)$.
    Note that here, if we require $\mathsf{Dec}'$ to make strictly 2 queries, we can arbitrarily have $\mathsf{Dec}'$ query one (or two) random codeword bits and simply ignore them, or query $j$ or $k$ a second time.
    \item Otherwise, calculate $j'$ (resp., $k'$) such that $\ell_{j'} = j$ (resp., $\ell_{k'} = k$) for $\ell_{j'} \in \overline{\mathcal{L}}$ (resp., $\ell_{k'} \in \overline{\mathcal{L}}$).
    $\mathsf{Dec}'$ then answers query $j$ (resp., $k$) with $y'_{j'}$ (resp., $y'_{k'}$).
    \item Finally, $\mathsf{Dec}'$ outputs what $\mathsf{Dec}$ outputs.
\end{itemize}
\end{remark}

\deleted[id=R1, comment={See comment number 31}]{%
    We note that the local decoder $\mathsf{Dec}$ for $C$ can also be used as a local decoder for $C_{J|\boldsymbol{\rho}}$, while preserving all the parameters. 
This is because $\mathsf{Dec}$ never needs to read a codeword bit which has become a constant function under the restriction $J|\boldsymbol{\rho}$.} 

For a subset $S \subseteq [m]$, we say ``$\mathsf{Dec}(i,\cdot)$ reads $S$'' if the event ``$j \in S$ and $k \in S$'' occurs where $j, k \in [m]$ are the queries made by $\mathsf{Dec}(i,\cdot)$. 
The following lemma says that conditioned on $\mathsf{Dec}(i,\cdot)$ \replaced[id=R3, comment={(12)}]{reading }{reads }some subset $S$, there is a way of modifying the bits in $S$ that flips the output of the decoder.

\begin{lemma} \label{lem:conditional-zero}
    Let $S \subseteq [m]$ be a subset such that $\Pr[\mathsf{Dec}(i,\cdot)\text{ reads }S]>0$. 
    \comment[id=R3]{(13) Author comment about what bit $b$ is conceptually.}
    \replaced{Let $\mathbf{s} \in \{0,1\}^{m}$ and $b \in \{0,1\}$ such that $\Pr[\mathsf{Dec}(i,\mathbf{s})=b ~|~ \mathsf{Dec}(i,\cdot)\text{ reads }S] > 0$. Then, there exists a string $\mathbf{z}\in \{0,1\}^{m}$ such that $\mathbf{z} \neq \mathbf{s}$, $\mathbf{z}_{[m]\setminus S} = \mathbf{s}_{[m] \setminus S}$, and}{%
    Then for any string $\mathbf{s} \in \{0,1\}^m$ and any bit $b \in \{0,1\}$, there exists a string $\mathbf{z} \in \{0,1\}^m$ such that $\mathbf{z}_{[m]\setminus S}=\mathbf{s}_{[m]\setminus S}$ and }
    \begin{align*}
    	\Pr\left[ \mathsf{Dec}(i, \mathbf{z}) = 1-b \mid \mathsf{Dec}(i,\cdot)\text{ reads }S \right] = 1.
    \end{align*}
\end{lemma}
\begin{proof}
    Let $\mathbf{x} \in \{0,1\}^{n}$ be a string with $x_i=1-b$. 
    Let $\mathbf{z} \in \{0,1\}^{m}$ be the string satisfying 
    \begin{equation*}
        \mathbf{z}_{S} = C(\mathbf{x})_{S}, \quad \mathbf{z}_{[m]\setminus S}=\mathbf{s}_{[m]\setminus S}. 
    \end{equation*}
    Since $\mathsf{Dec}$ has perfect completeness, we have
    \begin{equation*}
    	1 = \Pr\left[ \mathsf{Dec}(i, C(\mathbf{x})) = x_i \mid \mathsf{Dec}(i,\cdot)\text{ reads }S \right] = \Pr\left[ \mathsf{Dec}(i, \mathbf{z}) = 1-b \mid \mathsf{Dec}(i,\cdot)\text{ reads }S \right].
    \end{equation*}
    Moreover, it must be the case that $\mathbf{z}_{S} \neq \mathbf{s}_{S}$ since otherwise it would imply 
    \begin{align*}
        \Pr[\mathsf{Dec}(i, \mathbf{s}) = b ~|~ \mathsf{Dec}(i, \cdot)\text{ reads }S] = 0,
    \end{align*}
    which is a contradiction.
\end{proof}

The next lemma is a key step in our proof. 
It roughly says that there is a local decoder for $x_i$ in the standard sense as long as the size of $S_{i}$ is not too large.
\begin{lemma} \label{lem:LDC-reduction}
	Suppose $i \in [n]$ is such that $|S_{i}| \le \delta m/2$. 
    Then there is a $(2,\delta/2,1/2+\eps)$-local decoder $D_i$ for $i$. 
    In other words, for any $\mathbf{x} \in \{0,1\}^{n}$ and $\mathbf{y} \in \{0,1\}^{m}$ such that $\mathsf{HAM}(C(\mathbf{x}), \mathbf{y}) \le \delta m/2$, we have
	\begin{align*}
	    \Pr\left[ D_i(\mathbf{y}) = x_i \right] \ge \frac{1}{2} + \eps,
	\end{align*}
	and $D_i$ makes at most 2 queries into $\mathbf{y}$.
\end{lemma}
\begin{proof}
	Let $i \in [n]$ be such that $|S_{i}| \le \delta m/2$.
	The local decoder $D_i$ works as follows. 
    Given $\mathbf{x} \in \{0,1\}^{n}$ and $\mathbf{y} \in \{0,1\}^{m}$ such that $\mathsf{HAM}(C(\mathbf{x}), \mathbf{y}) \le \delta m/2$, $D_i$ obtains a query set $Q$ according to the query distribution of $\mathsf{Dec}(i,\cdot)$ conditioned on $Q \subseteq [m]\setminus S_i$. 
    Then, $D_i$ finishes by outputting the result returned by $\mathsf{Dec}(i,\cdot)$. 
	
    Denote by $E_i$ the event ``$\mathsf{Dec}(i,\cdot)$ reads $[m]\setminus S_i$'', \ie, both \deleted[id=R2]{two }queries made by $\mathsf{Dec}(i,\cdot)$ lie outside $S_i$. 
    In order for the conditional distribution to be well-defined, we need to argue that $E_i$ occurs with non-zero probability. 
    Suppose this is not the case; \added{that is, $\Pr[E_{i}] = 0$, and the event ``$j \in S_{i}$ or $k \in S_{i}$'' holds with probability $1$, where $j,k \in Q$.}
    \replaced{This implies }{meaning that }$Q \cap S_i \neq \emptyset$ for all possible query set\added[id=R3, comment={(14)}]{s} $Q \subseteq [m]$. 
    Let $\mathbf{z} \in \{0,1\}^m$ be the string obtained by applying \cref{lem:conditional-zero} with $S = S_i$, $\mathbf{s}=C(\mathbf{x})$ and $b=x_i$; 
    \added{in particular, $\mathbf{z}_{S_{i}} \neq C(\mathbf{x})_{S_{i}}$.}


    \added{Now, the contrapositive of \Cref{clm:bot-fix} tells us that if $\{ j,k \} \in F_{i}^{\geq 1}$ (i.e., the decoder can output $\bot$), then $j,k \in S_{i}$, which implies $Q \subseteq S_{i}$. In this case, the output of $\mathsf{Dec}(i, \mathbf{z})$ only depends on $\mathbf{z}_{S_{i}}$. Therefore, by \Cref{lem:conditional-zero}, we have} 
    \begin{align*}
        \Pr[\mathsf{Dec}(i, \mathbf{z}) = 1-b = 1-x_{i}] = 1.%
    \end{align*}%
    However, we also have that $\mathsf{HAM}(\mathbf{z}, C(\mathbf{x})) \leq |S_{i}| \leq \delta m$ by construction. Then, the relaxed decoding property  of $\mathsf{Dec}$ implies
    \begin{align*}
        \Pr[\mathsf{Dec}(i, \mathbf{z}) \in \{ x_{i}, \bot  \}] \geq \frac{1}{2}+\epsilon > 0,
    \end{align*}%
    which is a contradiction.

    \added{Now, suppose that $\{ j,k \} \notin F_{i}^{\geq 1}$, which implies $\{ j,k \} \in F_{i}^{0}$. By assumption, $j\in S_{i}$ or $k \in S_{i}$. Without loss of generality, suppose that $j \in S_{i}$. Then, \Cref{clm:fixable-or-useless} tells us that either (1) $k \in S_{i}$, (2) $C_{j} = x_{i}$, or (3) $C_{j} = 1-x_{i}$. In any case, we again have that the output of $\mathsf{Dec}(i,\mathbf{z})$ depends only on $\mathbf{z}_{S_{i}}$. For condition (1) this is clear, but conditions (2) and (3) imply that the query $k$ is a ``dummy'' query that the decoder can simply discard; thus, for conditions (2) and (3), we can treat the event ``$\mathsf{Dec}(i, \cdot)$ reads $S_{i}$'' as true; or, equivalently, the decoder simply sets $k$ to be some other query in $S_{i}$, e.g., $k = j$, since the decoder's output only depends on query $j$. Now again, by choice of $\mathbf{z}$, we have a contradiction via \Cref{lem:conditional-zero} since
    \begin{align*}%
        \frac{1}{2} + \eps \le \Pr\left[ \mathsf{Dec}(i,\mathbf{z}) \in \{x_i, \perp\} \right] = \Pr\left[ \mathsf{Dec}(i,\mathbf{z}) \in \{x_i, \perp\} \mid \text{ reads } S_{i}\right] = 0,
    \end{align*}%
    where the first inequality follows by the relaxed decoding property of $\mathsf{Dec}$ since $\mathsf{HAM}(C(\mathbf{x}),\mathbf{z}) = \mathsf{HAM}(C(\mathbf{x})_{S_{i}}, \mathbf{z}_{S_{i}})\le |S_i| \leq \delta m$. Therefore, we have shown that $\Pr[E_{i}] > 0$.}
	
	By definition of $D_i$, it makes at most 2 queries into $\mathbf{y}$. 
    Its success rate is given by
	\begin{align*}
        \Pr[D_i(\mathbf{y}) = x_i] = \Pr[\mathsf{Dec}(i,\mathbf{y}) = x_i \mid E_i].
	\end{align*}
	Therefore, it remains to show that
	\begin{align*}
        \Pr\left[ \mathsf{Dec}(i,\mathbf{y}) = x_i \mid E_i \right] \ge \frac{1}{2} + \eps.
	\end{align*} 
	
    Let $\mathbf{z}$ be the string obtained by applying \cref{lem:conditional-zero} with $S=S_i$, $\mathbf{s}=\mathbf{y}$ and $b=x_i$. 
    From previous discussions we see that conditioned on $\overline{E_i}$ (\ie, the event $E_i$ does not occur), the output of $\mathsf{Dec}(i,\mathbf{z})$ only depends on $\mathbf{z}_{S_{i}}$. 
    Therefore,
	\begin{equation}
		\Pr\left[ \mathsf{Dec}(i, \mathbf{z}) \in \{x_i, \perp\} \mid \overline{E_i} \right] = 1-\Pr\left[ \mathsf{Dec}(i, \mathbf{z}) = 1-x_i \mid \overline{E_i} \right] = 0.
		\label{eqn:conditional-zero}
	\end{equation}
	We also have that $\mathbf{z}$ is $\delta$-close to $C(\mathbf{x})$ since
	\begin{equation*}
	    \mathsf{HAM}(\mathbf{z}, C(\mathbf{x})) \le \mathsf{HAM}(\mathbf{z}, \mathbf{y}) + \mathsf{HAM}(\mathbf{y}, C(\mathbf{x})) \le |S_i| + \delta m/2 \le \delta m.
	\end{equation*}
	Thus, the relaxed decoding property of $\mathsf{Dec}$ gives
	\begin{align*}
	    \Pr\left[ \mathsf{Dec}(i, \mathbf{z}) \in \{x_i, \perp\} \right] \ge \frac{1}{2} + \eps.
	\end{align*}
	On the other hand, we also have
    \comment[id=R3]{(16) Display changed below for single-line equations.}
	\begin{align*}
        &\Pr\left[ \mathsf{Dec}(i, \mathbf{z}) \in \{x_i, \perp\} \right] &&\\
	    =& \Pr\left[ \mathsf{Dec}(i, \mathbf{z}) \in \{x_i, \perp\} \mid \overline{E_i} \right] \cdot \Pr\left[ \overline{E_i} \right] + \Pr\left[ \mathsf{Dec}(i, \mathbf{z}) \in \{x_i, \perp\} \mid E_i \right] \cdot \Pr\left[ E_i \right] \\
	    =&\Pr[ \mathsf{Dec}(i, \mathbf{z}) \in \{x_i, \perp\}|\overline{E_i}]\Pr[ \overline{E_i}] + \Pr[ \mathsf{Dec}(i, \mathbf{y}) \in \{x_i, \perp\}|E_i ]\Pr[E_i] \tag*{($\mathbf{z}[[m]\setminus S_i]=\mathbf{y}[[m]\setminus S_i]$)}\\
        =& \Pr\left[ \mathsf{Dec}(i, \mathbf{y}) \in \{x_i, \perp\} \mid E_i \right] \cdot \Pr\left[ E_i \right]  \tag*{(\Cref{eqn:conditional-zero})}\\
	    \le& \Pr\left[ \mathsf{Dec}(i, \mathbf{y}) \in \{x_i, \perp\} \mid E_i \right].
	\end{align*}
    Note that by \Cref{clm:bot-fix}, conditioned on $E_i$ \added{(i.e., $j,k \notin S_{i}$ occurs for $j,k \in Q$)}, $\mathsf{Dec}(i,\cdot)$ never outputs ``$\perp$''. 
    We, thus have 
	\begin{equation*}
	    \Pr\left[ \mathsf{Dec}(i, \mathbf{y}) = x_i \mid E_i \right] \ge \frac{1}{2} + \eps.\qedhere
	\end{equation*}
\end{proof}

\deleted{We remark once again that the above lemma holds for the restricted code $C_{J|\boldsymbol{\rho}}$, with $S_i$ replaced by $S_i'$.} 
Recall the definitions of $S_{i}$ and $T_{j}$, and define $w_{j} = |T_{j}|$ for all $j \in [m]$.
Define 
\begin{align}
    W \coloneqq \{ j \in [m] \colon |T_{j}| \geq 3\ln(8/\delta) \}.\label{eq:w-set}
\end{align}

\deleted{The lemma below will be useful later in the proof.} 
\added{Then, the following lemma shows that a }
\deleted{It shows that a} constant fraction of the message bits can be fixed so that most codeword bits $C_j$ with large $w_j$ become constants.

\begin{lemma} \label{lem:random-restriction}
    There exist a set $J \subseteq [n]$ and \replaced[id=R3, comment={(7)}]{an assignment }{assignments }$\boldsymbol{\rho} \in \{0,1\}^{\overline{J}}$ such that $|J| \ge n/6$ and $|W \setminus A| \le \delta m/4$, where $A \subseteq W$ \replaced[id=R3, comment={(7)}]{is the set of }{collects }all codeword bits $j \in W$ such that $C_j \restriction_{\mathbf{x}_{\overline{J}}=\boldsymbol{\rho}}$ is a constant function.
\end{lemma}
\begin{proof}
    Let $J$ be a random subset formed by selecting each $i \in [n]$ independently with probability $1/3$. 
    For each $i \in \overline{J}$, set $\rho_i = 0$ or $\boldsymbol{\rho}_i = 1$ with probability $1/2$. 
    \replaced{Then, we }{We }have $\Expect[|J|]=n/3$, and hence the Chernoff bound shows that $|J| < n/6$ with probability $\exp(-\Omega(n))$. 
    Furthermore, for each $j \in W$, $C_j\restriction_{\mathbf{x}_{\overline{J}}=\boldsymbol{\rho}}$ becomes a constant function except with probability $\delta/8$. 
    This is because for each $i \in T_j$, \added[id=R3, comment={(8)}]{at least one of} $C_j\restriction_{x_i=0}$ or $C_j\restriction_{x_i=1}$ is a constant function.
    \added[id=R1]{Thus, the probability that $C_{j} \restriction_{x_{i} = \rho_{i}}$ is constant is the probability $i \in T_{j}$ times the probability $x_{i}=\rho_{i}$ fixes $C_{j}$, which happens with probability at least $2/3 \cdot 1/2 = 1/3$ (note that if $C_{j}$ is constant for both $x_{i} \in \{0,1\}$, this probability becomes $2/3 \geq 1/3$).
    Moreover, $C_{j}$ becomes a constant \emph{for any} $i \in T_{j}$ and $\rho_{i}$ such that $C_{j}$ is fixed by $x_{i}=\rho_{i}$}
    Therefore, \added{for $C_{j}\restriction_{\mathbf{x}_{\bar{J}}=\boldsymbol{\rho}}$ to not be constant, $x_{i} = \rho_{i}$ must not fix $C_{j}$ \emph{for every $i \in T_{j}$}. This gives us}
    \begin{align*}
        \Pr\left[ C_j\restriction_{\mathbf{x}_{\overline{J}}=\boldsymbol{\rho}}\text{ is not constant} \right]\le \left(1-\frac{1}{3}\right)^{|T_j|} < e^{-|T_j|/3} \le \frac{\delta}{8}, 
    \end{align*} 
    where the last inequality follows by definition of $j \in W$ in \cref{eq:w-set}, which states $|T_j| \ge 3\ln(8/\delta)$.

    By linearity of expectation and Markov's inequality, we have
    \begin{align*}
        &\Pr\left[ \sum_{j \in W}\mathbf{1}\{C_j\restriction_{\mathbf{x}_{\overline{J}}=\boldsymbol{\rho}}\text{ is not constant}\} \ge \frac{\delta}{4}|W| \right] \\
        \le& \frac{ \Expect\left[ \sum_{j \in W}\mathbf{1}\{C_j\restriction_{\mathbf{x}_{\overline{J}}=\boldsymbol{\rho}}\text{ is not constant}\} \right]}{\delta|W|/4} \\
        =& \frac{\sum_{j \in W}\Pr\left[ C_j\restriction_{\mathbf{x}_{\overline{J}}=\boldsymbol{\rho}}\text{ is not constant} \right]}{\delta|W|/4} \\
    	\le& \frac{\delta/8 \cdot |W|}{\delta|W|/4} \le \frac{1}{2}.
    \end{align*}
    Applying a union bound gives
    \begin{align*}
        \Pr\left[ (|J| < n/6) \lor \left(\sum_{j \in W}\mathbf{1}\{C_j\restriction_{\*x_{\overline{J}}=\boldsymbol{\rho}}\text{ is not constant}\} \ge \frac{\delta}{4}|W|\right) \right] \le \exp(-\Omega(n))+ \frac{1}{2} < 1.
    \end{align*}
    Finally, we can conclude that there exist $J \subseteq [n]$ and $\boldsymbol{\rho}\in\{0,1\}^{\overline{J}}$ such that $|J| \ge n/6$, and $C_j\restriction_{\mathbf{x}_{\overline{J}}=\boldsymbol{\rho}}$ becomes a constant function for all but $\delta/4$ fraction of $j \in W$. 
\end{proof}

Let $J \subseteq [n]$ and $\boldsymbol{\rho} \in \{0,1\}^{\overline{J}}$ be given by \deleted[id=R3,comment={(9)}]{the }\Cref{lem:random-restriction}, and consider the restricted code $C_{J|\boldsymbol{\rho}}$. 
By rearranging the \replaced{input }{codeword } bits, we may assume $J=[n']$ where $n'=|J| \ge n/6$. 
Let $A \subseteq [m]$ be the set of codeword bits which get fixed to constants under $J|\boldsymbol{\rho}$, and for any $i \in [n]$, define $S_{i}' \coloneqq S_{i} \setminus A$.
%
We emphasize that $S_i'$ does not necessarily contain all codeword bits fixable by $x_i$ 
in the restricted code $C_{J|\boldsymbol{\rho}}$, as fixing some message bits may cause more codeword bits to be fixable by $x_i$.

Now, we prove our exponential lower bound for non-adaptive 2-query Hamming RLDCs via the following proposition.
\begin{proposition}\label{prop:non-adaptive-2qRLDC}
	Let $C \colon \{0,1\}^n \rightarrow \{0,1\}^m$ be a non-adaptive weak $(2,\delta,1/2+\eps)$-RLDC. 
    Then $m = 2^{\Omega_{\delta,\eps}(n)}$.
\end{proposition}
\begin{proof}
    Let $C_{J|\boldsymbol{\rho}}\colon \{0,1\}^{n'}\rightarrow \{0,1\}^{m}$ be the restricted code \added{such that the first $m'$ output bits are not fixed to constants under $\mathbf{x}_{\overline{J}} = \boldsymbol{\rho}$, and the last $m-m'$ bits are fixed to constants}, where $J|\boldsymbol{\rho}$ is given by \Cref{lem:random-restriction}, and let $A \subseteq [m]$ be the set of codeword bits which get fixed to constants \added{(i.e., $A = \{ m'+1, \dotsc, m \} $)}. 
    \added{Let $W \subseteq [m]$ be defined as in \cref{eq:w-set} with respect to the code $C$.}
    Define $S_i'\coloneqq S_i\setminus A$, $S_{i,-}'=S_{i,-}\setminus A$, $S_{i,+}'=S_{i,+}\setminus A$\added{, and $W' = W \setminus A$, where $S_{i,+} = W \cap S_{i}$ and $S_{i,-} = \overline{W} \cap S_{i}$.}
    \added{Note that $W'$ and $S_{i}'$ are subsets of all fixable codeword bits in the restricted code $C_{J|\boldsymbol{\rho}}$, and thus can also be viewed as $W', S_{i}' \subseteq [m']$. Similarly, $S_{i,-}', S_{i,+}'$ can be viewed as subsets of $[m']$ since $S_{i,-}' \cup S_{i,+}' = S_{i}'$.}
    
    Denote $T_j'\coloneqq \{i \in [n'] : j \in S_i'\}$. Since $S_i'\subseteq S_i$ for each $i$, we also have $T_j'\subseteq T_j$ for each $j$. 
    In particular, for each $j \notin W' \subseteq W$, we have $|T_j'| \le |T_j| \le 3\ln(8/\delta)$. 
    Therefore,
	\begin{equation*}
        \Expect_{i \in [n']}
	    \left[|S_{i,-}'|\right] = \frac{1}{n'}\sum_{i=1}^{n'}|S_{i,-}'| = \frac{1}{n'}\sum_{j \in [m']\setminus W'}|T_j'| \le 3\ln(8/\delta)\cdot \frac{m'}{n'}.
	\end{equation*}
	By Markov's inequality, we have
	\begin{equation*}
        \Pr_{i \in [n']}\left[|S_{i,-}'| > \delta m'/4 \right] \le \frac{12\ln(8/\delta)}{\delta n'} = O_{\delta}\left(\frac{1}{n'}\right).
	\end{equation*}
	In other words, there exists $I \subseteq [n']$ of size $|I| \ge n'-O_{\delta}(1)$ such that $|S_{i,-}'| \le \delta m'/4 \leq \delta m / 4$ for all $i \in I$. 
    Now, by \cref{lem:random-restriction}, we have $|S_{i,+}'| = |W' \setminus A| \leq |W'| = |W\setminus A| \leq \delta m / 4$.
    Therefore, for any $i \in I$, we have $|S_i'| = |S_{i,-}'| + |S_{i,+}'| \le \delta m'/4 + \delta m/4 \leq \delta m / 4 + \delta m / 4 = \delta m/2$. 
    \added{Using \cref{rem:restricted-decoder} to obtain decoder $\mathsf{Dec}'$ for $C_{J|\boldsymbol{\rho}}$ and by }applying \Cref{lem:LDC-reduction} on the restricted code $C_{J|\boldsymbol{\rho}}$ with the set $S_{i}'$, we can view $C_{J|\boldsymbol{\rho}}$ as a $(2,\delta/2,1/2+\eps)$-LDC for message bits in $I$ (for instance, we can arbitrarily fix the message bits outside $I$), where $|I| \geq n'-O_{\delta}(1) = \Omega_{\delta}(n)$. 
    Finally, applying \Cref{thm:two-query-lb} yields $m \geq 2^{cn'-1} = 2^{\Omega_{\delta,\varepsilon}(n)}$, as desired.
\end{proof}

\subsection{Lower bounds for adaptive 2-Query Hamming RLDCs}\label{subsec:adaptive-2qRLDC}

Now we turn to the actual proof, which still works (in spirit) for possibly adaptive decoders. 
\added[id=R1]{The main technical challenge with adaptive decoders is exactly adaptivity: with non-adaptive decoders, the queries $j,k$ are completely independent of the received word $\mathbf{y}$ (though they may be sampled from a distribution that depends on the received index $i$). For adaptive decoders, the first query is non-adaptive in the sense that it does not depend on the received word $\mathbf{y}$. However, the second query $k$ may now be defined over a distribution which depends explicitly on the sampled first query $j$ and the received bit $y_{j}$. Thus, to handle adaptive decoders, we must carefully handle these differences.}

Let $C$ be a \deleted{weak }$(2,\delta,1/2+\eps)$-RLDC with perfect completeness. 
We fix a relaxed (adaptive) decoder $\mathsf{Dec}$ for $C$. 
Without loss of generality, we assume $\mathsf{Dec}$ works as follows: on input $i \in [n]$, $\mathsf{Dec}(i,\cdot)$ picks the first query $j \in [m]$ according to a distribution $\mathcal{D}_i$. 
Let $b \in \{0,1\}$ be the answer to this query. 
Then $\mathsf{Dec}$ picks the second query $k \in [m]$ according to a distribution $\mathcal{D}_{i;j,b}$, and obtains an answer $b' \in \{0,1\}$. 
\deleted{Finally, $\mathsf{Dec}$ outputs a random variable $X_{i;j,b,k,b'}\in \{0,1,\perp\}$.} 
\added{Finally, $\mathsf{Dec}$ outputs an answer $c$ sampled from a distribution $\cD_{i;j,b,k,b'} \subseteq \{ 0,1,\bot \} $. We let $X_{i;j,b,k,b'}$ denote the random variable representing the output $c$ of the adaptive decoder.}

%
%
%
%
%
%
%


We partition the support of $\mathcal{D}_i$ into the following two sets:
\comment[id=R2]{Fixed subscript of $\mathcal{D}$ below.}
\begin{align*}
    F_i^{0} &:= \{j \in \mathrm{supp}(\mathcal{D}_i) : \forall b,b' \in \{0,1\}, k \in \mathrm{supp}(\mathcal{D}_{i;j,b}), \Pr[X_{i;j,b,k,b'} = \perp] = 0\}, \\
    F_i^{>0} &:= \{j \in \mathrm{supp}(\mathcal{D}_i) \colon \exists b,b' \in \{0,1\}, k \in \mathrm{supp}(\mathcal{D}_{i;j,b}), \Pr[X_{i;j,b,k,b'}=\perp] > 0\}.
\end{align*}

We will still apply the restriction guaranteed by \Cref{lem:random-restriction} to $C$. 
The sets $S_i$, $T_j$, $W$, $S_{i,-}$, $S_{i,+}$ (\replaced[id=R2]{and }{are }their counterparts for $C_{J|\rho}$) are defined in the exact same way.

The following claim is adapted from \Cref{clm:bot-fix}.
\begin{claim} \label{clm:adaptive-bot-fix}
    $(\mathrm{supp}(\mathcal{D}_i) \setminus S_i) \subseteq F_i^{0}$. 
\end{claim}
\begin{proof}
    Let $j \in \mathrm{supp}(\mathcal{D}_i) \setminus S_i$.
    We will show $j \in F_i^{0}$. 
    By the definition of $S_i$, $j \notin S_i$ means that there are partial assignments $\boldsymbol{\sigma}_{00}, \boldsymbol{\sigma}_{01}, \boldsymbol{\sigma}_{10}, \boldsymbol{\sigma}_{11} \in \{0,1\}^{n-1}$ such that 
    \comment[id=R3]{(18) Check consistency of $\sigma$ notation here and in other claims.}
	\begin{align*}
        C_j(\mathbf{x}_{-i} = \boldsymbol{\sigma}_{00}, x_i = 0) = 0, \quad C_j(\mathbf{x}_{-i} = \boldsymbol{\sigma}_{01}, x_i = 1) = 0, \\
        C_j(\mathbf{x}_{-i} = \boldsymbol{\sigma}_{10}, x_i = 0) = 1, \quad C_j(\mathbf{x}_{-i} = \boldsymbol{\sigma}_{11}, x_i = 1) = 1,
	\end{align*}
    where $\mathbf{x}_{-i}$ is defined as $(x_t : t \in [n]\setminus\{i\})$.
	
	Let $C_{00}, C_{01}, C_{10}, C_{11}$ be encodings of the corresponding assignments mentioned above. 
    Consider an arbitrary query $k \in \mathrm{supp}(\mathcal{D}_{i;j,0})$, and let $b_1', b_2'$ be the $k$-th bits of $C_{00}$ and $C_{01}$, respectively. 
    We note that $X_{i;j,0,k,b_1'}$ is the output of $\mathsf{Dec}(i,C_{00})$ conditioned on the queries $(j,0)$ and $(k,b_{1}')$, and $X_{i;j,0,k,b_2'}$ is the output of $\mathsf{Dec}(i,C_{01})$ conditioned on the queries $(j, 0)$ and $(k,b_{2}')$. 
    Due to perfect completeness of $\mathsf{Dec}$, we have
	\begin{align*}
	    \Pr[X_{i;j,0,k,b_1'}=0] = 1, \quad \Pr[X_{i;j,0,k,b_2'}=1]=1.
	\end{align*}
	Therefore, it must be the case that $b_1' \neq b_2'$, which implies that $\Pr[X_{i;j,0,k,b'}=\perp]=0$ for any $b' \in \{0,1\}$.
    An identical argument shows that $\Pr[X_{i;j,1,k,b'}=\perp]=0$ for any $k \in \mathrm{supp}(\mathcal{D}_{i;j,1})$ and $b' \in \{0,1\}$. 
    Thus, we have shown $j \in F_i^{0}$.
\end{proof}

We remark that the above claim also implies $F_i^{>0} \subseteq S_i$, since $\mathrm{supp}(\mathcal{D}_i)$ is a disjoint union of $F_i^{0}$ and $F_i^{>0}$. 
In other words, conditioned on the event that the first query $j$ is not contained in $S_i$, the decoder never outputs $\perp$.

\added{Our next goal is to show that the decoder, when conditioned on receiving \emph{honest codewords}, has a certain query structure when the first query $j$ is into a fixable codeword bit. Towards this, we can view the decoding process with respect to a fixed query string $\mathbf{y} \in \{0,1\}^{m}$ as a distribution over \emph{transcripts} of the form:
\begin{align*}
    \{ (i,j,b,k,b') \in [n] \times [m] \times \{0,1\} \times [m] \times \{0,1\} \}.
\end{align*}
A transcript $(i,j,b,k,b')$ is \emph{admissible} if $j \in \mathsf{supp}(\cD_{i})$ and $k \in \mathsf{supp}(\cD_{i;j,b})$. 
Let $\mathsf{AT}$ denote the set of all admissible transcripts.
We are interested in a stronger notion of admissible, which we call \emph{honest admissible}.
We say $(i,j,b,k,b')$ is \emph{honest admissible} if (1) it is admissible (i.e., $(i,j,b,k,b') \in \mathsf{AT}$), and (2) there exists $\mathbf{x} \in \{0,1\}^{n}$ such that $C_{j}(\mathbf{x}) = b$ and $C_{k}(\mathbf{x}) = b'$.
For admissible transcript $(i,j,b,k,b') \in \mathsf{AT}$, let
\begin{align*}
    \mathsf{HAT}(i,j,b,k,b') \coloneqq \{ \mathbf{x} \in \{0,1\}^{n} \colon C_{j}(\mathbf{x}) = b \wedge C_{k}(\mathbf{x}) = b' \}.
\end{align*}
That is, for a given transcript $(i,j,b,k,b')$, $\mathsf{HAT}(i,j,b,k,b')$ is the set of all inputs $\mathbf{x} \in \{0,1\}^{n}$ for which there exists a codeword $C(\mathbf{x})$ realizing the admissible transcript $(i,j,b,k,b')$.
Now, equivalently, a transcript $(i,j,b,k,b') \in \mathsf{AT}$ is honest admissible iff $\mathsf{HAT}(i,j,b,k,b') \neq \emptyset$.}

\added{Given honest admissible transcripts, we now prove the following claim.}
\begin{claim}\label{clm:adaptive-fixable-or-useless}
    \added{Let $j \in \mathrm{supp}(\cD_{i}) \cap S_{i}$ and suppose $C_{j}\restriction_{x_{i}=\gamma} \equiv b$ for fixed $\gamma,b \in \{0,1\}$.
    If $k \in \mathrm{supp}(\cD_{i;j,b}) \setminus S_{i}$, then all the following hold.}
    \begin{enumerate}
        \item \added{$\forall b' \in \{0,1\}, \mathsf{HAT}(i,j,b,k,b') \neq \emptyset$; }
        \item \added{$\forall b' \in \{0,1\}\colon\Pr[X_{i;j,b,k,b'} = \gamma]= 1$; and}
        \item \added{$C_{j}\restriction_{x_{i}=1-\gamma}\equiv 1-b$.
            Moreover, if $k' \in \mathrm{supp}(\cD_{i;j,1-b})\setminus S_{i}$, then (1) and (2) hold for transcripts $(i,j,1-b,k',b'')$ for all $b'' \in \{0,1\}$.}
    \end{enumerate}
\end{claim}
\begin{proof}
    \added{%
    Fix $i \in [n]$, $j \in \mathrm{supp}(\cD_{i})$, and $C_{j}\restriction_{x_{i}=\gamma} \equiv b$ for $\gamma,b \in \{0,1\}$.
    Suppose $k \in \mathrm{supp}(\cD_{i;j,b}) \setminus S_{i}$.
    This implies there exists two distinct partial assignments $\boldsymbol{\sigma}_{00}, \boldsymbol{\sigma}_{10} \in \{0,1\}^{n-1}$ such that 
    \begin{align*}
        C_{k}(\mathbf{x}_{-i} = \boldsymbol{\sigma}_{00}, x_{i} = \gamma) = 0 \qquad C_{k}(\mathbf{x}_{-i} = \boldsymbol{\sigma}_{10}, x_{i} = \gamma) = 1.
    \end{align*}
    Let $\mathbf{x}_{00}, \mathbf{x}_{10}$ denote these vectors; namely $(\mathbf{x}_{00})_{-i} = \boldsymbol{\sigma}_{00}$ and $(\mathbf{x}_{10})_{i} = \gamma$ (and similarly for $\mathbf{x}_{10}$).
    Now, by definition, we have $\mathbf{x}_{00} \in \mathsf{HAT}(i,j,b,k,0)$ and $\mathbf{x}_{10} \in \mathsf{HAT}(i,j,b,k,1)$, showing (1).}

    \added{Let $C_{00} = C(\mathbf{x}_{00})$ and $C_{10} = C(\mathbf{x}_{10})$ denote the encodings of these vectors.
    Then, by perfect correctness of the decoder, for honest admissible transcript $(i,j,b,k,b')$ for $b' \in \{0,1\}$, it must be the case that
    \begin{align*}
        \Pr[X_{i;j,b,k,b'} = \gamma] = 1.
    \end{align*}
    This shows (2).}

    \added{Now, suppose by way of contradiction that $C_{j}\restriction_{x_{i} = 1-\gamma} \not\equiv 1-b$.
    This implies that there are two distinct partial assignments $\boldsymbol{\sigma}_{01}, \boldsymbol{\sigma}_{11} \in \{0,1\}^{n-1}$ such that
    \begin{align*}
        C_{j}(\mathbf{x}_{-i} = \boldsymbol{\sigma}_{01}, x_{i} = 1-\gamma) = b \qquad C_{j}(\mathbf{x}_{-i} = \boldsymbol{\sigma}_{11}, x_{i} = 1-\gamma) = 1-b.
    \end{align*}
    Let $\mathbf{x}_{01}, \mathbf{x}_{11}$ denote these full assignments, and let $C_{01} = C_{01}(\mathbf{x}_{01})$ and $C_{11} = C_{11}(\mathbf{x}_{11})$ denote their encodings.
    Suppose that $b'' = (C_{01})_{k}$.
    By perfect completeness of the decoder, when given query access to $C_{01}$ and input $i$, we have
    \begin{equation*}
        \Pr[X_{i;j,b,k,b''} = 1-\gamma] = 1.
    \end{equation*}
    However, we have just established for all $b' \in \{0,1\}$ that
    \begin{equation*}
        \Pr[X_{i;j,b,k,b'} = \gamma] = 1.
    \end{equation*}
    Hence, a contradiction, which implies that $C_{j}\restriction_{x_{i}=1-\gamma} \equiv 1-b$.
    Finally, for $k' \in \mathrm{supp}(\cD_{i;j,1-b}) \setminus S_{i}$, we can again establish (1) and (2) for transcripts $(i,j,1-b,k',b'')$ for all $b'' \in \{0,1\}$ by the previous arguments in the proof.}
\end{proof}

We observe that \Cref{lem:conditional-zero,lem:random-restriction} both hold with respect to adaptive decoders; in particular, the proof of \Cref{lem:conditional-zero} makes no assumptions on how the decoder makes queries, and \Cref{lem:random-restriction} does not rely on the decoder.
\added{However, we need to adapt \Cref{lem:LDC-reduction} to the adaptive setting.
\Cref{clm:adaptive-bot-fix,clm:adaptive-fixable-or-useless,lem:conditional-zero} allow us to do so, and we capture this in the following lemma.
\begin{lemma}\label{lem:ldc-adaptive-reduction}
    Let $i \in [n]$ such that $|S_{i}| \leq \delta m / 2$.
    Then, there exists a $(2, \delta/2, 1/2+\varepsilon)$-adaptive local decoder $D_{i}$ for $i$.
    In other words, for any $\mathbf{x} \in \{0,1\}^{n}$ and $\mathbf{y} \in \{0,1\}^{m}$ such that $\mathsf{HAM}(C(\mathbf{x}), \mathbf{y}) \leq \delta m / 2$, we have
    \begin{equation*}
        \Pr[D_{i}(\mathbf{y}) = x_{i}] \geq \frac{1}{2} + \varepsilon,
    \end{equation*}
    and $D_{i}$ makes at most $2$ adaptive queries into $\mathbf{y}$.
\end{lemma}}
{
\begin{proof}
    Let $i \in [n]$ such that $|S_{i}| \leq \delta m / 2$.
    For all $j \in S_{i}$, let $(\gamma_{j},c_{j}) \in \{0,1\}^{2}$ be bits such that $C_{j}\restriction_{x_{i} = \gamma_{j}} = c_{j}$; if both $C_{j}\restriction_{x_{i} = 0}$ and $C_{j}\restriction_{x_{i}=1}$ are constant functions, set $\gamma_{j} = 0$ and $c_{j} = C_{j}\restriction_{x_{i}=0}$ without loss of generality.

    Define adaptive local decoder $D_{i}$ as follows: on input $\mathbf{y} \in \{0,1\}^{m}$, $D_{i}(\mathbf{y})$ does the following.
    \begin{enumerate}
        \item Sample query index $j \gets \cD_{i}$ according to $\mathsf{Dec}(i, \mathbf{y})$.
        \item If $j \in S_{i}$, output $\gamma_{j}$.
        \item Query and obtain $b = y_{j}$.
        \item Sample query index $k \gets \cD_{i;j,b}$ according to $\mathsf{Dec}(i, \mathbf{y})$.
        \item If $k \in S_{i}$, output $\gamma_{k}$.
        \item Query and obtain $b' = y_{k}$.
        \item Sample and output $\widetilde{b} \gets X_{i;j,b,k,b'}$.
    \end{enumerate}

    Now fix $\mathbf{x} \in \{0,1\}^{n}$ and $\mathbf{y} \in \{0,1\}^{m}$ such that $\mathsf{HAM}(C(\mathbf{x}), \mathbf{y}) \leq \delta m / 2$.
    Our goal is to show that for any such $\mathbf{y}$, we have $\Pr[D_{i}(\mathbf{y}) = x_{i}] \geq 1/2 + \varepsilon$.
    Let $\mathbf{x}' \in \{0,1\}^{n}$ be any vector such that $x_{i}' = 1-x_{i}$.
    Define new vector $\mathbf{z} \in \{0,1\}^{m}$ as $\mathbf{z}_{S_{i}} = C(\mathbf{x}')_{S_{i}}$ and $\mathbf{z}_{[m]\setminus S_{i}} = \mathbf{y}_{[m]\setminus S_{i}}$.
    Now, by the triangle inequality, observe that
    \begin{equation*}
        \mathsf{HAM}(\mathbf{z}, C(\mathbf{x})) \leq \mathsf{HAM}(\mathbf{z}, \mathbf{y}) + \mathsf{HAM}(\mathbf{y}, C(\mathbf{x})) \leq |S_{i}| + \frac{\delta m}{2} \leq \delta m.
    \end{equation*}
    Therefore, by the relaxed decoding property of $\mathsf{Dec}$, we have
    \begin{equation*}
        \Pr[\mathsf{Dec}(i, \mathbf{z}) \in \{ x_{i}, \bot \} ] \geq \frac{1}{2} + \varepsilon.
    \end{equation*}

    Next, we will show that
    \begin{equation*}
        \Pr[D_{i}(\mathbf{y}) = x_{i}] \geq \Pr[\mathsf{Dec}(i, \mathbf{z}) \in \{ x_{i}, \bot \} ].
    \end{equation*}
    We can expand the left-hand side to get the exact success probability of $D_{i}$ as
    \begin{align*}
        \Pr[D_{i}(\mathbf{y}) = x_{i}] &= \Pr[j \in S_{i}] \cdot \Pr[D_{i}(\mathbf{y}) = x_{i} ~|~ j \in S_{i}]\\
        &+ \Pr[j \notin S_{i}] \cdot \Pr[k \in S_{i}] \cdot \Pr[D_{i}(\mathbf{y}) = x_{i} ~|~ j \notin S_{i} \wedge k \in S_{i}]\\
        &+ \Pr[j \notin S_{i}] \cdot \Pr[k \notin S_{i}] \cdot \Pr[D_{i}(\mathbf{y}) = x_{i} ~|~ j,k \notin S_{i}].
    \end{align*}
    We now examine each conditional probability.
    First, we show that 
    \begin{equation}
        \Pr[D_{i}(\mathbf{y}) = x_{i} ~|~ j \in S_{i}] \geq \Pr[\mathsf{Dec}(i, \mathbf{z}) \in \{ x_{i}, \bot \} ~|~ j \in S_{i}].\label{eq:lem4-cond1}
    \end{equation}
    Notice that by definition of $D_{i}$, for any fixed $J \in \mathrm{supp}(\mathcal{D}_{i}) \cap S_{i}$, we have
    \begin{equation*}
        \Pr[D_{i}(\mathbf{y}) = x_{i} ~|~ j \in S_{i} \wedge J = j] = 
        \begin{cases}
            1 & x_{i} = \gamma_{J}\\
            0 & x_{i} = 1-\gamma_{J}
        \end{cases}.
    \end{equation*}
    If $x_{i} = \gamma_{J}$, then notice $\Pr[\mathsf{Dec}(i,\mathbf{z}) \in \{ x_{i}, \bot \} ~|~ j \in S_{i} \wedge J = j] \leq 1$ by definition of probability.
    Now, if $x_{i} = 1-\gamma_{J}$, we must show that $\Pr[\mathsf{Dec}(i,\mathbf{z}) \in \{ x_{i}, \bot \} ~|~ j \in S_{i} \wedge J = j] = 0$.
    In this case, we have $\gamma_{J} = 1-x_{i} = x_{i}'$ from the previously defined $x_{i}'$.
    Expanding the conditional probability we are interested in, we have
    \begin{align*}
        &\Pr[\mathsf{Dec}(i, \mathbf{z}) \in \{ x_{i}, \bot \} ~|~ j \in S_{i} \wedge J = j]\\ 
        =& \Pr[k \in S_{i} ~|~ J = j] \cdot \Pr[\mathsf{Dec}(i, \mathbf{z}) \in \{ x_{i}, \bot \} ~|~ j, k \in S_{i} \wedge J = j]\\
        +& \Pr[k \notin S_{i} ~|~ J = j] \cdot \Pr[\mathsf{Dec}(i, \mathbf{z}) \in \{ x_{i}, \bot \} ~|~ j \in S_{i}, k \notin S_{i} \wedge J = j].
    \end{align*}
    For the first conditional probability, we have $j,k \in S_{i}$.
    By construction of $\mathbf{z}$, we have $\mathbf{z}_{S_{i}} = C(\mathbf{x}')_{S_{i}}$, where $x_{i}' = 1-x_{i} = \gamma_{J}$.
    By perfect correctness of the decoder, this implies that for $b = y_{J}$ and any $K \in \mathrm{supp}(\cD_{i;J,b}) \cap S_{i}$, we have
    \begin{align*}
        &\Pr[\mathsf{Dec}(i,\mathbf{z}) \in \{ x_{i} , \bot \} ~|~ j,k \in S_{i} \wedge J = j \wedge K = k]\\
        =&\Pr[\mathsf{Dec}(i,C(\mathbf{x}')) \in \{ x_{i} , \bot \} ~|~ j,k \in S_{i} \wedge J = j \wedge K = k] = 0,
    \end{align*}
    where the last equality holds since perfect completeness gives us
    \begin{align*}
        \Pr[\mathsf{Dec}(i,C(\mathbf{x}')) = 1-x_{i} ~|~ j,k \in S_{i} \wedge J = j \wedge K = k] = 1,
    \end{align*}
    as required.

    Now, for the second conditional probability, recall that by definition of $\gamma_{J}$, we have $C_{J}\restriction_{x_{i} = \gamma_{J}} \equiv c_{J}$.
    By definition of $\mathbf{x}'$, we have $x_{i}' = \gamma_{J} = 1-x_{i}$.
    This implies that $z_{J} = C(\mathbf{x}')_{J} = c_{J}$, which implies that the second query of $\mathsf{Dec}(i, \mathbf{z})$ conditioned on $j \in S_{i}$ and $J = j$ satisfies $k \in \mathsf{supp}(\mathcal{D}_{i;J,c_{J}})$.
    For this case, suppose $k \notin S_{i}$ and fix any $K \in \mathrm{supp}(\mathcal{D}_{i;J,c_{J}}) \setminus S_{i}$.
    Then, by Case (2) of \cref{clm:adaptive-fixable-or-useless}, we have that for any $z_{K} \in \{ 0,1 \} $, it holds that $\mathsf{Dec}(i, \mathbf{z}) = \gamma_{J}$ with probability $1$.
    Since $\gamma_{J} = 1-x_{i}$, we again have that $\mathsf{Dec}(i,\mathbf{z}) \in \{ x_{i}, \bot \} $ occurs with probability $0$ for any $J \in \mathrm{supp}(\mathcal{D}_{i}) \cap S_{i}$ and $K \in \mathrm{supp}(\mathcal{D}_{i;J,c_{J}}) \setminus S_{i}$.
    Thus, we have shown \cref{eq:lem4-cond1}.

    Next, we show that
    \begin{equation}
        \Pr[D_{i}(\mathbf{y}) = x_{i} ~|~ j \notin S_{i} \wedge k \in S_{i}] \geq \Pr[\mathsf{Dec}(i, \mathbf{z}) \in \{ x_{i}, \bot \} ~|~ j \notin S_{i} \wedge k \in S_{i}].\label{eq:lem4-cond2}
    \end{equation}
    First, fix $J \in \mathrm{supp}(\mathcal{D}_{i}) \setminus S_{i}$ and $b = y_{J}$. 
    Then, notice that $z_{J} = y_{J} = b$ for some $b \in \{0,1\}$ by construction of $\mathbf{z}$.
    Thus, both $D_{i}(\mathbf{y})$ and $\mathsf{Dec}(i,\mathbf{z})$ sample $k \gets \cD_{i;j,b}$ (the same distribution).
    Now, fix $K \in \mathrm{supp}(\mathcal{D}_{i;J,b}) \cap S_{i}$ and condition on $K = k$.
    Then for $D_{i}(\mathbf{y})$, when $K\in S_{i}$, by construction the decoder outputs $\gamma_{K}$, so again we have
    \begin{equation*}
        \Pr[D_{i}(\mathbf{y}) = x_{i} ~|~ j \notin S_{i} \wedge k \in S_{i} \wedge J = j \wedge K = k] = 
        \begin{cases}
            1 & x_{i} = \gamma_{K}\\
            0 & x_{i} = 1-\gamma_{K}
        \end{cases}.
    \end{equation*}
    As before, if $x_{i} = \gamma_{K}$, it holds that $1 \geq \Pr[\mathsf{Dec}(i, \mathbf{z}) \in \{ x_{i}, \bot \} ~|~ j \notin S_{i} \wedge k \in S_{i} \wedge J = j \wedge K = k]$.
    Now, suppose that $x_{i} = 1-\gamma_{K}$.
    We must show that
    \begin{equation*}
        \Pr[\mathsf{Dec}(i, \mathbf{z}) \in \{ x_{i}, \bot \} ~|~ j \notin S_{i} \wedge k \in S_{i} \wedge J = j \wedge K = k] = 0.
    \end{equation*}
    By \cref{clm:adaptive-bot-fix}, $J \notin S_{i}$ implies that $\mathsf{Dec}(i, \mathbf{z})$ never outputs $\bot$, so we only need to argue
    \begin{equation*}
        \Pr[\mathsf{Dec}(i, \mathbf{z}) = x_{i} ~|~ j \notin S_{i} \wedge k \in S_{i} \wedge J = j \wedge K = k] = 0.
    \end{equation*}
    Since $K \in S_{i}$, by definition of $\gamma_{K}$, we have $C_{K}\restriction_{x_{i} = \gamma_{K}} \equiv c_{K}$.
    Again, by construction of $\mathbf{x}'$, we have $x_{i}' = 1-x_{i} = \gamma_{K}$, so $C(\mathbf{x}')_{K} = c_{K} = z_{K}$.
    Now, since $J \notin S_{i}$, there exists $\mathbf{x}'' \neq \mathbf{x}'$ such that $x_{i}'' = x_{i}'$ and $C_{J}(\mathbf{x}'') = b$.
    This implies that $(i,J,b,K,c_{K})$ is an honestly admissible transcript.
    Therefore, by perfect completeness of $\mathsf{Dec}$, we have
    \begin{equation*}
        \Pr[X_{i;J,b,K,c_{K}} = \gamma_{j} = 1-x_{i} ~|~ j \notin S_{i} \wedge k \in S_{i} \wedge J = j \wedge K = k] = 1,
    \end{equation*}
    which shows that 
    \begin{equation*}
        \Pr[\mathsf{Dec}(i, \mathbf{z}) = x_{i} ~|~ j \notin S_{i} \wedge k \in S_{i} \wedge J = j \wedge K = k] = 0.
    \end{equation*}
    Therefore, \cref{eq:lem4-cond2} holds.

    Finally, we show that
    \begin{equation}
        \Pr[D_{i}(\mathbf{y}) = x_{i} ~|~ j,k \notin S_{i}] \geq \Pr[\mathsf{Dec}(i, \mathbf{z}) \in \{ x_{i}, \bot \} ~|~ j,k \notin S_{i}].\label{eq:lem4-cond3}
    \end{equation}
    Again, fix $J \in \mathrm{supp}(\mathcal{D}_{i}) \setminus S_{i}$, $b = y_{J}$, and $K \in \mathrm{supp}(\mathcal{D}_{i;J,b}) \setminus S_{i}$, and condition on $J = j \wedge K = k$.
    First, again since $J \notin S_{i}$, the decoder $\mathsf{Dec}$ never outputs $\bot$ by \cref{clm:adaptive-bot-fix}.
    Now, by definition of $D_{i}$, we have
    \begin{equation*}
        \Pr[D_{i}(\mathbf{y}) ~|~ j,k \notin S_{i} \wedge J = j \wedge K = k] = \Pr[\mathsf{Dec}(i, \mathbf{y}) = x_{i} ~|~ j, k \notin S_{i} \wedge J = j \wedge K = k].
    \end{equation*}
    Moreover, since $\mathbf{z}_{[m] \setminus S_{i}} = \mathbf{y}_{[m] \setminus S_{i}}$, we have
    \begin{equation*}
        \Pr[\mathsf{Dec}(i, \mathbf{y}) = x_{i} ~|~ j,k \notin S_{i} \wedge J = j \wedge K = k] = \Pr[\mathsf{Dec}(i, \mathbf{z}) = x_{i} ~|~ j,k \notin S_{i} \wedge J = j \wedge K = k],
    \end{equation*}
    which shows \cref{eq:lem4-cond3}.

    All together, \crefrange{eq:lem4-cond1}{eq:lem4-cond3} show that
    \begin{align*}
        \Pr[D_{i}(\mathbf{y}) = x_{i}] &= \Pr[j \in S_{i}] \cdot \Pr[D_{i}(\mathbf{y}) = x_{i} ~|~ j \in S_{i}]\\
        &+ \Pr[j \notin S_{i}] \cdot \Pr[k \in S_{i}] \cdot \Pr[D_{i}(\mathbf{y}) = x_{i} ~|~ j \notin S_{i} \wedge k \in S_{i}]\\
        &+ \Pr[j \notin S_{i}] \cdot \Pr[k \notin S_{i}] \cdot \Pr[D_{i}(\mathbf{y}) = x_{i} ~|~ j,k \notin S_{i}]\\
        &\geq \Pr[j \in S_{i}] \cdot \Pr[\mathsf{Dec}(i, \mathbf{z}) \in \{ x_{i}, \bot \} ~|~ j \in S_{i} ]\\
        &+ \Pr[j \notin S_{i}] \cdot \Pr[k \in S_{i}] \cdot \Pr[\mathsf{Dec}(i,\mathbf{z}) \in \{ x_{i},\bot \}  ~|~ j \notin S_{i} \wedge k \in S_{i}]\\
        &+ \Pr[j \notin S_{i}] \cdot \Pr[k \notin S_{i}] \cdot \Pr[\mathsf{Dec}(i,\mathbf{z}) = \{ x_{i}, \bot \}  ~|~ j,k \notin S_{i}]\\
        &=\Pr[\mathsf{Dec}(i, \mathbf{z}) \in \{ x_{i}, \bot \} ] \geq \frac{1}{2} + \varepsilon.\qedhere
    \end{align*}
\end{proof}}

\begin{remark}[\Cref{lem:LDC-reduction} vs. \Cref{lem:ldc-adaptive-reduction}]
    Though the statements of \Cref{lem:LDC-reduction,lem:ldc-adaptive-reduction} are nearly identical, the proofs are quite different.
    One important but subtle difference is that the reduction in the proof of \cref{lem:LDC-reduction} \emph{preserves perfect completeness} of the non-adaptive decoder $\mathsf{Dec}$.
    However, the reduction in the proof of \Cref{lem:ldc-adaptive-reduction} does not preserve perfect completeness.
    Our original proof of \cref{lem:ldc-adaptive-reduction} also preserved perfect completeness, but relied on a previous version of \cref{clm:adaptive-fixable-or-useless} within which an anonymous reviewer found a subtle bug (we thank this anonymous reviewer for pointing this out).
    Our new version of \cref{lem:ldc-adaptive-reduction} addresses this bug, though it is not clear how to get a reduction which preserves perfect completeness; however, this is unnecessary since the 2-query LDC lower bound of Kerenidis and de Wolf \cite{KerenidisW04} (\cref{thm:two-query-lb}) also holds with respect to LDCs with \emph{imperfect completeness} error $1/2 + \varepsilon$.
\end{remark}

Finally, we are ready to prove \Cref{thm:main}, which we restate here for convenience. 


\begin{theorem}[\Cref{thm:main}, Restated]
    Let $C \colon \{ 0,1 \} ^{n} \rightarrow \{ 0,1 \} ^{m}$ be an adaptive $(2,\delta,1/2+\eps)$-RLDC.
    Then, $m=2^{\Omega_{\delta,\varepsilon}(n)}$.
\end{theorem}
\begin{proof}[Proof of \Cref{thm:main}]
    The proof is nearly identical to the proof of \Cref{prop:non-adaptive-2qRLDC}.
    Let $C_{J|\boldsymbol{\rho}}\colon \{0,1\}^{n'}\rightarrow \{0,1\}^{m}$ be the restricted code \added{such that the first $m'$ output bits are not fixed to constants under $\mathbf{x}_{\overline{J}} = \boldsymbol{\rho}$, and the last $m-m'$ bits are fixed to constants}, where $J|\boldsymbol{\rho}$ is given by \Cref{lem:random-restriction}, and let $A \subseteq [m]$ be the set of codeword bits which get fixed to constants \added{(i.e., $A = \{ m'+1, \dotsc, m \} $)}. 
    Let $W \subseteq [m]$ be defined as in \cref{eq:w-set} with respect to the code $C$.
    Again define $W' \coloneqq W \setminus A$, $S_{i}' \coloneqq S_{i} \setminus A$, $S_{i,-} \coloneqq \overline{W} \cap S_{i}$, $S_{i,+} \coloneqq W \cap S_{i}$, $S_{i,-}' \coloneqq S_{i,-}\setminus A$, and $S_{i,+}' \coloneqq S_{i,+} \setminus A$.
    As before, $W', S_{i}'$ are subsets of fixable codeword bits in $C_{J|\boldsymbol{\rho}}$, and thus can be viewed as subsets of $[m']$, along with $S_{i,-}', S_{i,+}'$.
    Finally, define $T_{j}' \coloneqq \{ i \in [n']\colon j \in S_{i}' \} $.

    Now, as with the proof of \Cref{prop:non-adaptive-2qRLDC}, we have $|T'_{j}| \leq 3 \ln(8/\delta)$ and thus can show that
    \begin{align*}
        \Pr_{i\in[n']} \left[ |S_{i,-}'| > \delta m' / 4 \right]  = O_{\delta}\left( \frac{1}{n'} \right).
    \end{align*}
    This implies there exists $I \subseteq[n']$ of size $|I| \geq n' - O_{\delta}(1)$ such that $|S_{i,-}'| \leq \delta m' / 4$ for all $i \in I$.
    By \Cref{lem:random-restriction}, we have $|S_{i,+}'| = |W' \setminus A| \leq |W'| = |W \setminus A| \leq \delta m / 4$.
    Therefore, for any $i \in I$, we have $|S_{i}'| = |S_{i,-}'| + |S_{i,+}'| \leq \delta m' / 4 + \delta m / 4 \leq \delta m / 4 + \delta m / 4 =  \delta m / 2$.

    We can now apply \cref{lem:ldc-adaptive-reduction} on the restricted code $C_{J|\boldsymbol{\rho}}$ with the set $S_{i}'$ and can view $C_{J|\boldsymbol{\rho}}$ as a $(2, \delta/2, 1/2+\varepsilon)$-adaptive LDC for message bits in $I$, noting that $|I| = \Omega_{\delta}(n)$.
    Applying \cref{thm:two-query-lb} yields $m \geq 2^{cn'-1} = 2^{\Omega_{\delta,\varepsilon}(n)}$, which completes the proof.
\end{proof}

\section*{Acknowledgements}
The authors graciously thank several anonymous reviewers for their helpful and insightful comments to improve this manuscript.
We especially thank a review for pointing out that the random restriction process may introduce new fixable codeword bits in the set $T_{j}$ (cf. \cref{foot:tj-restriction}), and thank another reviewer for finding an error with a counter example in the previous version of \cref{clm:adaptive-fixable-or-useless}.

\paragraph{AI Acknowledgements.}
ChatGPT was used to help refine \cref{clm:adaptive-fixable-or-useless} to avoid a counter-example highlighted by an anonymous reviewer.
In particular, several updated versions were submitted to ChatGPT and interactively refined.
ChatGPT was also used to help update the proof of \cref{lem:ldc-adaptive-reduction} as the old proof was no longer compatible with \cref{clm:adaptive-fixable-or-useless}; the main idea of the reduction was generated by ChatGPT, and the authors clarified and expanded the proof to its current form.
All AI-assisted material was reviewed, revised, and verified by the authors, who determined the research contributions and technical content and take full responsibility for the final manuscript.

\bibliographystyle{alphaurl}
\bibliography{local}



\end{document}